\documentclass[11pt,a4paper]{article}

\usepackage[T1]{fontenc}
\usepackage{newpxtext}
\usepackage{mathpazo}
\usepackage{amsthm, amsfonts, amssymb, graphicx}
\usepackage{tikz}
\usepackage{blkarray}
\usepackage{dsfont}
\usepackage[margin=1truein]{geometry}
\usepackage[full]{complexity}
\usepackage[hypertexnames=false]{hyperref}
\usepackage{comment, bbm, enumitem}
\usepackage{xcolor}
\usepackage{pmat}
\usepackage{algorithm, algorithmicx}
\usepackage[noend]{algpseudocode}
\usepackage{underscore}
\usetikzlibrary{matrix}
\usepackage{amsmath} % Ensure amsmath is loaded after other conflicting packages
\usepackage[capitalize]{cleveref}
\usepackage{authblk}

\crefname{observation}{observation}{observations}
\Crefname{observation}{Observation}{Observations}

\newcommand{\F}{\mathbb{F}}
\newtheorem{theorem}{Theorem}[section]
\newtheorem{lemma}[theorem]{Lemma}
\newtheorem{claim}[theorem]{Claim}

\newtheorem{observation}[theorem]{Observation}

\DeclareMathOperator{\ad}{adj}
\newcommand{\adj}[1]{(#1)^{\ad}}

\newcommand\DE{\stackrel{\mathclap{\scalebox{.5}{\mbox{DS}}}}{=}}
\newcommand\dge{\stackrel{\mathclap{\scalebox{.5}{\mbox{DE}}}}{=}}
\newcommand\PME{\stackrel{\mathclap{\scalebox{.4}{\mbox{PME}}}}{=}}

\newcommand{\abs}[1]{\lvert #1 \rvert}

\algdef{SE}{PFor}{EndPFor}[1]{\textbf{for all} \(\mbox{#1}\) \textbf{do in parallel}}{\textbf{end for}}
\theoremstyle{definition}
\newtheorem{definition}[theorem]{Definition}
\newtheorem{remark}[theorem]{Remark}

\algnewcommand{\algorithmicinput}{\textbf{Input:}}
\algnewcommand{\Input}{\item[\algorithmicinput]}
\algnewcommand{\algorithmicoutput}{\textbf{Output:}}
\algnewcommand{\Output}{\item[\algorithmicoutput]}

\newenvironment{repstatement}[1]{\noindent {\bf \cref{#1} (restated).} \em}{}

\usepackage{nicematrix}
\usepackage{mathtools}

\DeclarePairedDelimiterX{\Set}[2]{\{}{\}}{\,{#1}\,:\,{#2}\,}

\DeclareMathOperator*{\tw}{ct}

\DeclareMathOperator*{\diag}{diag}

\DeclareMathOperator*{\rank}{rank}

\newcommand{\comp}[1]{\overline{#1}}

\newcommand{\sgc}[1]{}

\title{{Characterizing and Testing Principal Minor Equivalence of Matrices}}
\author[1]{Abhranil Chatterjee \thanks{\texttt{abhneil@gmail.com}}}
\author[2]{Sumanta Ghosh
\thanks{ \texttt{besusumanta@gmail.com}}}
\author[3]{Rohit Gurjar \thanks{\texttt{rgurjar@cse.iitb.ac.in}, supported by SERB MATRICS grant}}
\author[3]{Roshan Raj \thanks{ \texttt{roshanraj@cse.iitb.ac.in}}}

\affil[1]{Indian Statistical Institute, Kolkata}
\affil[2]{Chennai Mathematical Institute}
\affil[3]{Indian Institute of Technology Bombay}

\date{\today}

\begin{document}
\maketitle  
\begin{abstract}
    Two matrices are said to be principal minor equivalent if 
     they have equal corresponding principal minors of all orders.
    We give a characterization of principal minor equivalence 
    and a deterministic polynomial time algorithm to check if two given matrices are principal minor equivalent. 
    Earlier such results were known for certain special cases like symmetric matrices, skew-symmetric matrices with {0, 1, -1}-entries, and matrices with no cuts (i.e., for any non-trivial partition of the indices, the top right block or the bottom left block must have rank more than 1).
    
    As an immediate application, we get an algorithm to check if the determinantal point processes corresponding to two given kernel matrices (not necessarily symmetric) are the same. 
    As another application, we give a deterministic polynomial-time test to check equality of two multivariate polynomials, each computed by a symbolic determinant with a rank 1 constraint on coefficient matrices.
    \end{abstract}

\begin{comment}
    Earlier such results were known 
    for certain special cases like symmetric matrices, skew-symmetric matrices with $\{0, 1, -1\}$-entries, and matrices with no cuts (a subset of indices $X$ is a cut of an $n\times n$ matrix $A$ if $1 < |X| < n-1 $ and $\rank(A[X,\comp{X}])=\rank(A[\comp{X}, X]) = 1$)
%

    As an immediate application, we get an algorithm to check if the determinantal point processes corresponding to two given kernel matrices (not necessarily symmetric) are the same. 
    %
    As another application, we give a deterministic polynomial-time polynomial identity test for the following case: $\det(A_0 + \sum_{i=1}^m A_i x_i) \stackrel{?}{=} \det (B_0 + \sum_{i=1}^m B_i x_i)$,
    where $A_j, B_j$ are $n\times n$ matrices with $\rank(A_i) = \rank(B_i) = 1 $ for $1\leq i \leq m$. 
     
\end{comment}

\section{Introduction}
The determinant of a matrix is a fundamental object of study in mathematics that has found numerous applications throughout computer science, physics, and other fields.
A \emph{minor} of a matrix is the determinant of one of its square submatrices and its order is the size of the corresponding submatrix.
%deleting a set of rows and a set of columns. 
%
A \emph{principal minor} of a matrix is a minor  
obtained by deleting the same set of rows and columns.
%
%A  minor is said to be of \emph{order} $k$ if the corresponding submatrix is of size $k\times k$. 
Principal minors play an important role in a variety of applications,
for example,
convexity of functions and positive semidefinite matrices~\cite{BoydVandenberghe2004}, 
the linear complementarity problem and P-matrices~\cite{Murty1972},
counting number of forests via the Laplacian matrix~\cite{BapatSiva2011},
 and inverse eigenvalue problems~\cite{Friedland1977}.

In this paper, we investigate a basic question about principal minors -- what is the relationship between two  $n\times n$ matrices 
which have equal corresponding principal minors of all orders (i.e., two 
matrices $A$ and $B$ such that 
for all $S \subseteq \{1,2,\dots, n\}$, $\det(A[S,S]) =\det(B[S,S])$). 
We call two such matrices to be \emph{principal minor equivalent (PME)}.
Observe that two matrices are PME if and only if all their corresponding principal submatrices have the same set of eigenvalues.
We seek answers of the following two questions.
\begin{description}
\item[Question~1 (characterization).] Can we identify a property $\mathcal{P}$ such that two matrices are PME if and only if they satisfy $\mathcal{P}$?

\item[Question~2 (efficient algorithm).] Can we efficiently check whether two matrices are PME or not?
\end{description}

The question of characterizing the relationship between two PME matrices has been extensively studied \cite{ES80, Loewy86, Ahmadieh23, BoussairiChaichaaCherguiLakhlifi2021, BC16}.
One motivation for studying this question comes from the problem of 
learning determinantal point processes~\cite{KuleszaTaskar2012, Urschel2017, Brunel2018}
and the closely related principal minor assignment problem~\cite{GriffinTsatsomeros2006, RisingKuleszaTaskar2015, BrunelUrschel2024}. 
While our original motivation to study this question came from an application to the polynomial identity testing problem (see Section~\ref{sec:intro-applications}). 

To move towards characterizing PME matrices, 
let us first consider some trivial operations which preserve all the principal minors. 
Two matrices $A$ and $B$ are called \emph{diagonally similar} if there exists an invertible diagonal matrix $D$ such that $A = D B D^{-1}$.
We call two matrices $A$ and $B$ \emph{diagonally equivalent} if $A$ is diagonally similar to $B$ or $B^T$.
It is easy to verify that any two diagonally equivalent matrices are PME.
Interestingly, Engel and Schneider~\cite{ES80} showed that the converse is also true when one of the matrices is symmetric.
That is, principal minor equivalence of a symmetric matrix with another matrix implies their diagonal equivalence (in fact, diagonal similarity). 
As one can efficiently check whether two matrices are diagonally equivalent or not, it also yields an efficient algorithm to decide principal minor equivalence in this case.

In general, principal minor equivalence does not imply diagonal equivalence,
as demonstrated by the following example. 
Consider the following block diagonal matrix $A$ and a block upper triangular matrix $B$:
\begin{equation}
A = \begin{pmatrix}
A_1 &0\\
0 &A_2
\end{pmatrix},
\quad\quad
B = 
\begin{pmatrix}
A_1 &A_3\\
0 &A_2
\end{pmatrix}.
\end{equation}
It is easy to see that $A$ and $B$
are principal minor equivalent oblivious to the entries of $A_3$, but they are not diagonally equivalent. 
Such matrices 
	that can be written as a block upper triangular matrix by permuting some rows and corresponding columns are called \emph{reducible} matrices (and irreducible otherwise). 
 For any $n\times n$ matrix $A$, define a graph with the vertex set $[n]$ and allow an edge $(i,j)$ if and only if the $(i,j)$-th entry of $A$ is nonzero. 
 We can equivalently define reducible matrices as the ones whose graph has more than one \emph{strongly connected components}. 
One can show that two matrices are PME if and only if
they have the same set of irreducible blocks
and their corresponding irreducible blocks are PME (see, for example, \cite[Section 5]{Ahmadieh23}). 
Hence, we can restrict our focus to irreducible matrices. 

In a series of works, Hartfiel and Loewy \cite{Hartfiell84}, and Loewy \cite{Loewy86} extended the result of Engel and Schneider~\cite{ES80} to general irreducible matrices with no \emph{cuts}. 
An $n\times n$ matrix $A$ is said to have a cut $X\subseteq [n]$, if $2 \leq |X| \leq n-2$ and both the submatrices $A[X, \comp{X}]$ and $A[\comp{X}, X]$  have rank one (the submatrices cannot have rank zero if $A$ is irreducible).
They showed that for any irreducible matrix $A$ with no cuts and any matrix $B$, if $A$ and $B$ are PME, then $A$ and $B$ are also diagonally equivalent. 
So, the case which remained unclear was that of irreducible matrices with cuts. 
Engel and Schneider~\cite[Example 3.7]{ES80} had given an example of two $4 \times 4$ matrices which are PME, but not diagonally equivalent. 
Clearly, both these matrices must have a cut. 

\paragraph{The cut-transpose operation.}
Recently, Ahmadieh \cite[Lemma 4.5]{Ahmadieh23} gave a general recipe that for any irreducible matrix $A$ with a cut, finds another matrix $B$ that is PME to $A$, but not necessarily diagonally equivalent to $A$. 
For this they define an operation on matrices with a cut, which we 
refer as cut-transpose. 
Consider a matrix $A$ and let $X$ be a cut of $A$. 
From the definition of a cut, $A$ must be of the following form:
\[
A = 
\begin{pmatrix}
	M &pq^T\\
	uv^T &N
\end{pmatrix},
\]
where the submatrix $A(X,X) = M$ and $A(\comp{X}, \comp{X}) = N$
and $p, q, u, v$ are column vectors of appropriate dimensions.
Define a \emph{cut-transpose} operation on $A$ with respect to cut $X$, which transforms $A$ to a new matrix $\tw(A, X)$ as follows:
\[
\tw(A,X) = 
\begin{pmatrix}
	M &pu^T\\
	qv^T &N^T
\end{pmatrix}.
\]
Ahmadieh~\cite{Ahmadieh23} showed that cut-transpose is a principal minor preserving operation.
A natural conjecture would be that any two irreducible PME matrices are related by a sequence of cut-transpose operations. 
To elaborate, let us define any two matrices $A$ and $B$ as \emph{cut-transpose equivalent} if 
there is a sequence $A = A_0$, $A_1, \ldots,$ $A_k$ of matrices
such that for each $0 \leq i \leq k-1$, $A_{i+1} = \tw (A_{i}, X_{i})$ 
for some cut $X_{i}$ of $A_i$, and $A_k$ is diagonally equivalent to $B$.
Can one show that two irreducible matrices are PME if and only if
they are cut-transpose equivalent?    

Boussa\"{i}ri and Chergui~\cite{BC16}  gave a characterization for principal minor equivalent matrices for a special case,
when the two matrices are skew-symmetric with entries from $\{-1, 0, 1\}$
and all their off-diagonal entries in the first row are nonzero. 
Interestingly, this characterization turns out to be cut-transpose equivalence
with a restriction. 
Moreover, they conjectured that the characterization should be true for arbitrary skew-symmetric matrices. 
In a follow up work, Boussa\"{i}ri, Chaïcha\^{a}, Chergui, and Lakhlifi~\cite{BoussairiChaichaaCherguiLakhlifi2021}
proved a similar result for another special case called generalized tournament matrices 
(non-negative matrices $A$ with $A+A^T = J_n -I_n$, where $J_n$ is all ones matrix). 
The  two settings use a transformations called HL-clan-reversal and
clan-inversion, respectively, which coincide with some restrictions of the cut-transpose operation.
Both these work build on a combinatorial result~\cite{BoussairiIlleLopezThomasse2004} about directed graphs with a similar flavor.
The combinatorial result, in turn, is a generalization of Gallai's theorem~\cite{Gallai1967} 
which states that if two partially ordered sets 
have the same comparability graph, then 
they are related by a sequence of orientation reversal operations (see~\cite{BoussairiIlleLopezThomasse2004, Mohring1985}).
This orientation reversal on a poset is a special instance of cut-transpose on the corresponding
skew-symmetric matrix. 

This series of works strengthens the confidence in the conjecture that cut-transpose equivalence should be a characterization of PME for arbitrary irreducible matrices.
However, their techniques are graph-theoretic and it is not clear how they can be 
generalized to arbitrary matrices. 
We instead employ algebraic techniques and show that conjecture is indeed true, thereby completely resolving Question~1. 
This extends the results for above mentioned special cases and also proves the conjecture of Boussa\"{i}ri and Chergui~\cite{BC16} about skew-symmetric matrices.
Moreover, we show that for any two $n\times n$ irreducible PME matrices $A$ and $B$, the cut-transpose sequence contains at most $2n$ matrices.
\begin{theorem}
\label{thm:main-one}
Let $A$ and $B$ be two $n\times n$ irreducible matrices over any field.
Then,  $A$ and $B$ are principal minor equivalent if and only if there exists a sequence of $n\times n$ matrices $(A=A_0,A_1,\dots, A_k)$  with $k< 2n$ such that
\begin{equation}
\label{eqn:main-thm-one}
\text{ for } 0\leq i \leq k-1,\ A_{i+1} = \tw(A_{i},X_i) \text{ for some cut } X_i \text{ of }A_{i} 
\end{equation}
and $A_k$ is diagonally equivalent to $B$.
\end{theorem}

Now, let us come to the question of an efficient algorithm to check 
if two given matrices are PME (Question~2). 
If one is allowed the use of randomness, then there is a simple algorithm for this task via a reduction to polynomial identity testing. 
Consider a $n\times n$ diagonal matrix $Y$ with variables $y_1, y_2, \dots, y_n$ in the diagonal. 
Observe that two $n\times n$ matrices $A$ and $B$  are PME if and only if the following is a polynomial identity (i.e., coefficient-wise equality)
\[\det(A+Y) = \det(B+Y).\]
There is a simple randomized algorithm for polynomial identity testing: just plug-in some random numbers for the variables and then check the equality (see~\cite{Sch80,DL78,Zip79}). 
There is no deterministic polynomial time algorithm known for polynomial identity testing in general, but we can still ask if there is one for this special case. 
We answer this question positively.
Recall the earlier discussion about reducible matrices and note that
testing PME for two matrices reduces to the same for their corresponding
irreducible blocks. 
\begin{theorem}
\label{thm:main-two}
There exists a deterministic polynomial-time algorithm
that for any two given  $n\times n$ matrices $A$ and $B$ over any field, 
decides whether the corresponding principal minors of $A$ and $B$ are equal or not. 
If they are equal, then as a certificate, the algorithm outputs cut-transpose sequences
for the corresponding irreducible blocks of the two matrices as guaranteed by Theorem~\ref{thm:main-one}.
\end{theorem}

\subsection{Applications}
\label{sec:intro-applications}

\paragraph{Polynomial Identity Testing.} 
As mentioned earlier our motivation for the principal minor equivalence problem came from the polynomial identity testing (PIT) problem.
Given two multivariate polynomials in a succinct representation, the PIT problem asks to decide whether the two polynomials are identical (i.e., all corresponding coefficients are equal). 
One of the widely studied and useful representation for multivariate polynomials is the \emph{determinantal representation}.
We say that  a polynomial $f(x_1, \ldots, x_m)\in \mathbb{F}[x_1, \ldots, x_m]$ has a determinantal representation of size $n$ if there exists matrices $A_0, A_1, \ldots, A_m\in \mathbb{F}^{n\times n}$ such that $f = \det(A_0 + \sum A_ix_i)$. 
The determinantal representation is known to be almost as expressive as algebraic circuits (see~\cite{Valiant1979} for more details).
The PIT problem admits a randomized polynomial-time algorithm~\cite{Sch80,DL78,Zip79}. 
Obtaining a deterministic algorithm for PIT remains a challenging  open problem that would have interesting implications in proving lower bounds, and many other algorithmic applications (see, for example, \cite{ShpilkaY10}). 
Unable to solve it for the general setting, the problem has been studied for various restricted settings. 

One such restricted setting is symbolic determinant under rank one restriction.
Here we ask for testing whether $\det(A_0 + \sum_{i=1}^m A_ix_i) =0$, for given matrices $A_i$, where $\rank(A_i)=1$ for $1\leq i \leq m$. 
There has been a lot of interest in this particular setting because of its connections with some combinatorial optimization problems like bipartite matching
and linear matroid intersection (see~\cite{Edmonds1967, Lovasz1989, NarayananSaranVazirani1992}), and algebraic problems like maximum rank matrix completion (see~\cite{Ivanyos2010, Geelen1999, Murota1993}).
The connection with combinatorics also gives a deterministic polynomial time algorithm for identity testing in this setting. 
In fact, there is also an efficient blackbox PIT (quasi-polynomial time) known for this case~\cite{DBLP:conf/stoc/GurjarT17} (blackbox means that the algorithm cannot see the input, it can only evaluate the given polynomial at any point). 

When we have an efficient algorithm to test whether a given polynomial from a class is zero, the next natural question one can ask is to test whether two given polynomials from that class are equal. 
If the class of polynomials is closed under addition, the equality question easily reduces to testing zeroness of a given polynomial (from the same class). 
Many well studied classes of polynomials have this property, for example, sparse polynomials, bounded-depth circuits, constant fan-in depth-3 circuits etc.
On the other hand, there are classes like ROABPs, which are not closed under addition~\cite{KayalNairSaha2020}, and for which the equality testing question has been studied independently~\cite{GurjarKorwarSaxena2016}. 
Symbolic determinant with rank one restriction is another such class. 
To the best of our knowledge, the class is not known to be closed under addition. 
Given that zeroness testing is known for this class, 
a natural extension would be to ask  if two given polynomials from this class are equal.
To the best of our knowledge, no non-trivial (deterministic) algorithm was known for
testing equality of two polynomials from this class (symbolic determinant with rank one restriction).
We show that this problem reduces to testing principal minor equivalence, and hence,
has a deterministic polynomial-time algorithm.

\begin{theorem}
\label{thm:main-three}
There exists a deterministic polynomial time algorithm such that given two sequences of $n\times n$ matrices $(A_0,A_1,\ldots, A_m)$ and $(B_0,B_1,\ldots, B_m)$ over any field, with the rank of $A_i$ and $B_i$ being at most $1$ for $1\leq i\leq n$, it decides whether $\det(A_0+A_1y_1+\ldots+A_my_m)=\det(B_0+B_1y_1+\ldots+B_my_m)$.
\end{theorem}

\paragraph{Determinantal Point Processes.}
As mentioned earlier, one motivation to study principal minors come from 
determinantal point processes (DPP). 
DPP are a family of probabilistic models which originated in physics~\cite{Macchi1975}, 
and which has subsequently found a wide range of applications in machine learning~\cite{KuleszaTaskar2012}, 
for example, document summarization, recommender systems, information retrieval etc.\
 (see references given in \cite{GartrellBrunelDohmatobKrichene2019, Urschel2017}). 
 Conventionally, a DPP is defined using principal minors of
  an $n\times n$ symmetric positive semidefinite matrix $K$,
  called a kernel, whose eigenvalues are between 0 and 1. 
  The DPP corresponding to kernel matrix $K$ is a probability distribution on subsets $Y$ of $\{1,2,\dots, n\}$
  such that for any subset $J \subseteq \{1,2, \dots, n\}$,
  \[\Pr [J  \subseteq Y] = \det (K[J]),\]
  where $K[J]$ is the principal submatrix of $K$ corresponding to set $J$ (see \cite{Kulesza2012}). 
  DPPs are useful in settings where one needs to generate a diverse set of objects 
  (larger principal minor means the vectors associated with the subset span a larger volume). 

   Symmetric DPPs (as defined above with a symmetric kernel matrix) 
 have a significant expressive power, however they come with a limitation. 
 Symmetric DPPs can model only repulsive interactions. 
 That is, 
 any pair of items has a negative correlation -- selection of one item 
 reduces the chances of selection of another item. 
To overcome this limitation, 
nonsymmetric determinantal point process has been proposed,
that is, DPP with a nonsymmetric kernel matrix $K$.
A nonsymmetric kernel matrix can model both positive and negative correlations. 
Lately, there have been a few works on nonsymmetric DPPs~\cite{Brunel2018, GartrellBrunelDohmatobKrichene2019, Reddy2022, HanGartrell2022, Arnaud2024}. 
One of the crucial questions in the study of DPPs
%~\cite{GillenwaterKuleszaFoxTaskar2014, MarietSra2015, Urschel2017, GartrellPaquetKoenigstein2017, DupuyBach2018} 
is to understand how are two kernel matrices related which produce the same DPP, which was explicitly asked in some works on learning DPPs~\cite{Brunel2018, BrunelUrschel2024}.
This is precisely the principal minor equivalence problem.
While it was already understood in the case of symmetric DPPs, 
we answer it for nonsymmetric DPPs in this work. 
Theoerm~\ref{thm:main-one} gives a characterization of the set of matrices $K'$
such that $DPP(K') = DPP(K)$ for a given kernel matrix $K$ (not necessarily symmetric).
Theorem~\ref{thm:main-two} gives a deterministic polynomial time algorithm to test whether
two given kernel matrices will produce the same DPP. 

\subsection{Proof overview}
In this subsection, we give a short overview of the proof techniques involved in proving Theorems~\ref{thm:main-one}, \ref{thm:main-two}
and \ref{thm:main-three}.
We start with Theorem~\ref{thm:main-one} which characterizes principal minor equivalence of irreducible matrices by cut-transpose equivalence. 
We have already discussed that cut-transpose equivalence implies principal minor equivalence. 
Thus, only the other direction remains to be shown, i.e., principal minor equivalence implies cut-transpose equivalence. 

\paragraph{Reduction to the case of all nonzero entries.} Our proof of Theorem~\ref{thm:main-one} works with an assumption
that the matrices have all nonzero entries.
We reduce the general case to this case using a technique from earlier works~\cite{Loewy86, Hartfiell84, Ahmadieh23},
namely, the transformation $A \mapsto (A+Z)^{\mathrm{adj}}$ (or $(A+Z)^{-1}$),
where $Z$ is a diagonal matrix with entries as distinct algebraically
independent elements (or indeterminates).
They showed that for any irreducible matrix $A$, the matrix $(A+Z)^{\mathrm{adj}}$ has all nonzero entries. 
Moreover, two irreducible matrices $A$ and $B$ are PME if and only if 
$(A+Z)^{\mathrm{adj}}$ and $(B+Z)^{\mathrm{adj}}$ are. 
They have also shown that $A$ and 
$(A+Z)^{\mathrm{adj}}$ have the same set of cuts. 
In Lemma~\ref{lem:AdjugateTwistEquality}, 
we show that the cut-transpose operation commutes with operation 
$A \mapsto (A+Z)^{\mathrm{adj}}$. 
This means that matrices $A$ and $B$ are cut-transpose equivalent if and only if $(A+Z)^{\mathrm{adj}}$ and $(B+Z)^{\mathrm{adj}}$ are.
Hence, it is sufficient to prove 
Theorem~\ref{thm:main-one} for matrices with all nonzero entries. 

\paragraph{No common cuts.} 
To prove the characterization for matrices with cuts,
a natural strategy would be to somehow decompose the matrices along a chosen cut
and then argue inductively for the obtained smaller pairs of matrices.  
From Loewy's characterization~\cite{Loewy86}, it follows that for any two irreducible PME matrices $A$ and $B$,
if $A$ has a cut, then $B$ must also have one. 
However, it is not necessary a subset of indices which is a cut in matrix $A$, is also a cut in matrix $B$.
In fact, it is possible that the two matrices do not have even one cut in common. 
Following is such an example of two irreducible PME matrices. 
\[
A=
\begin{pmatrix}
     0 & 1 & 0 & 0 & 0 & 0 \\
     0 & 0 & 1 & 0 & 0 & 0 \\
     0 & 0 & 0 & 1 & 0 & 0 \\
     0 & 0 & 0 & 0 & 1 & 0 \\
     0 & 0 & 0 & 0 & 0 & 1 \\
     1 & 0 & 0 & 0 & 0 & 0 
\end{pmatrix}, 
\quad
B =
\begin{pmatrix}
     0 & 0 & 0 & 1 & 0 & 0 \\
     0 & 0 & 0 & 0 & 0 & 1 \\
     0 & 0 & 0 & 0 & 1 & 0 \\
     0 & 1 & 0 & 0 & 0 & 0 \\
     1 & 0 & 0 & 0 & 0 & 0 \\
     0 & 0 & 1 & 0 & 0 & 0 
\end{pmatrix}
.
\]
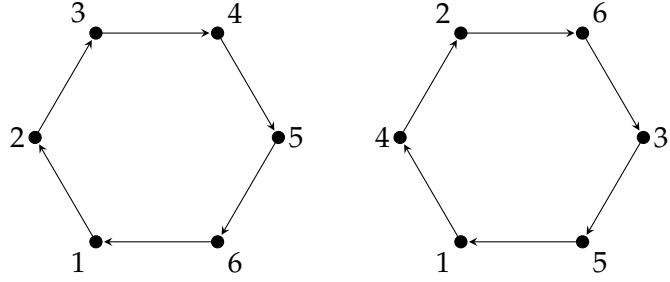
\begin{figure}
    \centering
\begin{tikzpicture}[scale=0.8]
    \draw [fill] (0,0) node [below left] {$1$} circle [radius=0.1];
    \draw [fill] (-1,1.73) node [left] {$2$} circle [radius=0.1];
    \draw [fill] (0,3.46) node [above left] {$3$} circle [radius=0.1];
    \draw [fill] (2, 3.46) node [above right] {$4$} circle [radius=0.1];
    \draw [fill] (3,1.73) node [right] {$5$} circle [radius=0.1];
    \draw [fill] (2,0) node [below right] {$6$} circle [radius=0.1];
    \draw [shorten >=3, -stealth] (0,0)--(-1,1.73);
    \draw [shorten >=3, -stealth] (-1,1.73)--(0,3.46);
    \draw [shorten >=3, -stealth] (0,3.46)--(2, 3.46);
    \draw [shorten >=3, -stealth] (2, 3.46)--(3,1.73);
    \draw [shorten >=3, -stealth] (3,1.73)--(2,0);
    \draw [shorten >=3, -stealth] (2,0)--(0,0); 

        \draw [fill] (6,0) node [below left] {$1$} circle [radius=0.1];
    \draw [fill] (5,1.73) node [left] {$4$} circle [radius=0.1];
    \draw [fill] (6,3.46) node [above left] {$2$} circle [radius=0.1];
    \draw [fill] (8, 3.46) node [above right] {$6$} circle [radius=0.1];
    \draw [fill] (9,1.73) node [right] {$3$} circle [radius=0.1];
    \draw [fill] (8,0) node [below right] {$5$} circle [radius=0.1];

    \draw [shorten >=3, -stealth] (6,0)--(5,1.73);
    \draw [shorten >=3, -stealth] (5,1.73)--(6,3.46);
    \draw [shorten >=3, -stealth] (6,3.46)--(8, 3.46);
    \draw [shorten >=3, -stealth] (8, 3.46)--(9,1.73);
    \draw [shorten >=3, -stealth] (9,1.73)--(8,0);
    \draw [shorten >=3, -stealth] (8,0)--(6,0); 

 \end{tikzpicture}

    \caption{Directed graphs associated with  two matrices $A$ and $B$.}
    \label{fig:twoCycles}
\end{figure}
To see that the  two matrices $A$ and $B$ are PME, 
recall that the determinant is a sum over cycle covers in 
the  directed graph associated with the matrix. 
And observe that  the associated directed graphs (Figure~\ref{fig:twoCycles})
have only one cycle, and thus, both matrices have only one nonzero principal minor (i.e., $\det(A)=\det(B) =1$). 
Now, to see that the two matrices do not have a common cut,
observe that for these matrices, any cut corresponds to
a path in associated directed graph. 
And there is no subset of vertices simultaneously forming a path in both the graphs.

We handle such cases with no common cuts by transforming one of the 
matrices to have a common cut with the other. 
Then we prove cut-transpose equivalence by induction based on the size of the matrices. 
The following points summarize our proof strategy.
\begin{enumerate}
    \item For any two irreducible PME matrices $A$ and $B$, 
    we show that the matrix $A$ has a cut in common either with 
    matrix $B$ or with another matrix $B'$ obtained from $B$ 
    via a cut-transpose operation.
    \item  Then assuming that the two given matrices have a common cut, 
    we ``decompose'' each matrix along a common cut to obtain two smaller matrices. 
    We argue that the two obtained pairs of matrices are also PME and hence, are cut-transpose equivalent by induction hypothesis.  
    Then we are able to lift their cut-transpose equivalence to the given matrices. 
    \item  The base case for the induction is $4 \times 4$ matrices.
    
\end{enumerate}
We now elaborate on each of the above points. 

\paragraph{Base case: $4 \times 4$ matrices} 
If $A$ and $B$ are  $4 \times 4$ irreducible PME matrices,
then we show (Lemma~\ref{lem:size4Matrix})  that (i) either the two are diagonally equivalent or (ii) they have a common cut and when we do a cut-transpose on matrix $A$ along the common cut, we get a matrix diagonally equivalent to $B$.
For $3\times 3$ or smaller matrices, there is no cut, and hence the
two matrices must be diagonally equivalent~\cite{Hartfiell84}. 

\paragraph{Getting a common cut.} 
To get a common cut in the given matrices $A$ and $B$, we consider a (inclusion-wise) minimal cut $S$ in $A$. 
We show that if $S$ is not a cut in $B$ then 
$S$ must have size two
(Lemma~\ref{lem:Greaterthan2commoncut}).
Moreover, in that case we can argue that 
there is a cut $X$ in matrix $B$ such that
when we apply cut-transpose on $B$ along the cut $X$, 
we get another matrix $B'$ where $S$ is a cut 
(Lemma~\ref{lem:cutsizetwo}). 
Clearly, proving cut-transpose equivalence between $A$ and $B'$ will imply the same between $A$ and $B$. 
The proofs of these two lemmas build on some other technical claims (Lemmas~\ref{thm:cutmap}, ~\ref{lem:size4Matrix}, \ref{lem:minCutGeneral}, \ref{lem:Blowerhalf}), and this is where most of the technical novelty lies.

Let us see how the matrix $B'$ is obtained in the example described above.
Observe that the matrix $A$ has a cut $S=\{1,2\}$, which
is not a cut in matrix $B$.
Let us consider the cut $X=\{1,4\}$ in matrix $B$ and apply cut-transpose along it.
We obtain the following matrix $B'$, which has $S=\{1,2\}$ as a cut, 
as desired. 
 Figure~\ref{fig:cut-transpose-cycle} shows the cut-transpose operation on the associated directed graph. 
\[
B =
\begin{pmatrix}
     0 & 0 & 0 & 1 & 0 & 0 \\
     0 & 0 & 0 & 0 & 0 & 1 \\
     0 & 0 & 0 & 0 & 1 & 0 \\
     0 & 1 & 0 & 0 & 0 & 0 \\
     1 & 0 & 0 & 0 & 0 & 0 \\
     0 & 0 & 1 & 0 & 0 & 0 
\end{pmatrix}
\quad
\xrightarrow[\text{along } \{1,4\}]{\text{cut-transpose}}
\quad
B' =
\begin{pmatrix}
     0 & 0 & 0 & 1 & 0 & 0 \\
     1 & 0 & 0 & 0 & 0 & 0 \\
     0 & 0 & 0 & 0 & 0 & 1 \\
     0 & 0 & 0 & 0 & 1 & 0 \\
     0 & 0 & 1 & 0 & 0 & 0 \\
     0 & 1 & 0 & 0 & 0 & 0 
\end{pmatrix}
\]

\begin{figure}
    \centering
\begin{tikzpicture}[scale=0.8]
    \draw [fill] (0,0) node [below left] {$1$} circle [radius=0.1];
    \draw [fill] (-1,1.73) node [left] {$4$} circle [radius=0.1];
    \draw [fill] (0,3.46) node [above left] {$2$} circle [radius=0.1];
    \draw [fill] (2, 3.46) node [above right] {$6$} circle [radius=0.1];
    \draw [fill] (3,1.73) node [right] {$3$} circle [radius=0.1];
    \draw [fill] (2,0) node [below right] {$5$} circle [radius=0.1];
    \draw [shorten >=3, -stealth] (0,0)--(-1,1.73);
    \draw [shorten >=3, -stealth] (-1,1.73)--(0,3.46);
    \draw [shorten >=3, -stealth] (0,3.46)--(2, 3.46);
    \draw [shorten >=3, -stealth] (2, 3.46)--(3,1.73);
    \draw [shorten >=3, -stealth] (3,1.73)--(2,0);
    \draw [shorten >=3, -stealth] (2,0)--(0,0); 

    \draw [dashed] (-1,3) to (1.5,-0.8);

    \draw[-stealth] (4,1.73)--(7,1.73);
    \node at (5.5,1.73) [above] {cut-transpose};
    \node at (5.5,1.73) [below] {along  $\{1,4\}$};

        \draw [fill] (9,0) node [below left] {$1$} circle [radius=0.1];
    \draw [fill] (8,1.73) node [left] {$4$} circle [radius=0.1];
    \draw [fill] (9,3.46) node [above left] {$5$} circle [radius=0.1];
    \draw [fill] (11, 3.46) node [above right] {$3$} circle [radius=0.1];
    \draw [fill] (12,1.73) node [right] {$6$} circle [radius=0.1];
    \draw [fill] (11,0) node [below right] {$2$} circle [radius=0.1];

    \draw [shorten >=3, -stealth] (9,0)--(8,1.73);
    \draw [shorten >=3, -stealth] (8,1.73)--(9,3.46);
    \draw [shorten >=3, -stealth] (9,3.46)--(11, 3.46);
    \draw [shorten >=3, -stealth] (11, 3.46)--(12,1.73);
    \draw [shorten >=3, -stealth] (12,1.73)--(11,0);
    \draw [shorten >=3, -stealth] (11,0)--(9,0); 

 \end{tikzpicture}

    \caption{Applying cut-transpose on the directed graph associated with matrix $B$.}
    \label{fig:cut-transpose-cycle}
\end{figure}
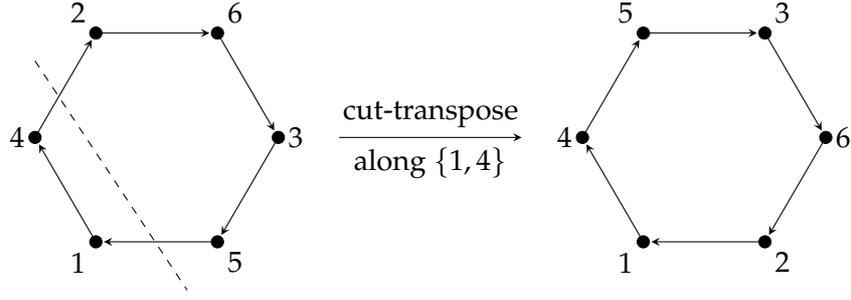

\paragraph{Decomposition into smaller matrices.}
One of our crucial ideas is to define the right decomposition
of a matrix along one of its cuts.
For an $n \times n$ matrix $A$ with a cut $S \subseteq [n]$, 
we consider a decomposition of $A$ into two matrices $A_1$
and $A_2$ defined as follows: choose two arbitrary indices
$s \in \comp{S}$ and $t \in S$, and define
$A_1 := A[S+s]$ and $A_2 := A[\comp{S}+t]$. 
Recall that we assume all off-diagonal entries to be nonzero, 
hence, the choice of $s$ and $t$ do not really matter. 
As discussed earlier, we can assume that there is set $S$, which is a 
cut in 
both the matrices $A$ and $B$.
We similarly decompose $B$ into matrices $B_1$ and $B_2$.
It is easy to see that if $A$ and $B$ are PME, then so 
are $A_i$ and $B_i$, for $i=1,2$. 

If $S$ is a minimal cut of $A$, then we show that $A_1$
has no cut (Lemma~\ref{lem:minCutGeneral}). 
In that case, $B_1$ also does not have a cut and is diagonally equivalent to $A_1$ (from Loewy's characterization). 
If it so happens that $A_2$ and $B_2$ are already diagonally equivalent, then we show 
that $A$ is either diagonally equivalent to $B$ or 
$\tw(B, \comp{S})$ (Lemma~\ref{lem:SplussToN}). 

The more interesting case is when $A_2$ and $B_2$ are not diagonally equivalent. 
Then by induction hypothesis, we assume that $A_2$ and $B_2$
are cut-transpose equivalent. 
In the final step in the proof, we show that we can lift the cut-transpose sequence that relates $A_2$ and $B_2$ to a cut-transpose sequence for $A$ and $B$ as described in \cref{cl:smallToBig}.
This lifting procedure is as follows: for each cut $X$ in the sequence,
we either replace it with $X \cup S$ or keep it as it is, depending on whether $X$ contains $t$ or not.
Finally we append $\comp{S}$ to the sequence (or do not append), depending on
whether $A_1$ is diagonally similar to $B_1^T$ or $B_1$.
We demonstrate this lifting of the cut-transpose sequence via an example. 
Consider two PME matrices 
\[
A=
\begin{pmatrix}
    1 & 3 & 1 & 1 & 1 \\
    2 & 1 & -1 & -1 & -1 \\ 
    1 & 2 &  2 & 1 & 1   \\
    2 & 4 &  -2 & 3 & 4   \\
    -1 & -2 & 1 & 5 & 6  \\
\end{pmatrix}
,
\quad 
B=
\begin{pmatrix}
    1 & 2 & 1 & 2 & -1 \\
    3 & 1 & 2 & 4 & -2 \\ 
    1 & -1 &  2 & 2 & -1   \\
    1 & -1 &  -1 & 3 & 5   \\
    1 & -1 & -1 & 4 & 6  \\
\end{pmatrix}
\]
Let us index the rows and columns of these two matrices by $\{a,b,c,d,e\}$.
Observe that matrices $A$ and $B$ have common cut $S= \{a,b\}$. 
We decompose each of them to obtain two smaller matrices as given below.
Here matrices $A_1$ and $B_1$ are submatrices of $A$ and $B$, respectively, indexed by $\{a,b,c\}$.
Similarly, matrices $A_2$ and $B_2$ are submatrices indexed by $\{b,c,d,e\}$. 
\[
A_1=
\begin{pmatrix}
    1 & 3 & 1  \\ 
    2 & 1 & -1  \\
    1 & 2 &  2 \\
\end{pmatrix}, \;
A_2=
\begin{pmatrix}
     1 & -1 & -1 & -1 \\ 
     2 &  2 & 1 & 1   \\
     4 &  -2 & 3 & 4   \\
     -2 & 1 & 5 & 6 
\end{pmatrix}, \;
B_1=
\begin{pmatrix}
    1 & 2 & 1  \\
    3 & 1 & 2  \\ 
    1 & -1 &  2    \\
\end{pmatrix}, \;
B_2=
\begin{pmatrix}
     1 & 2 & 4 & -2 \\ 
     -1 &  2 & 2 & -1   \\
    -1 &  -1 & 3 & 5   \\
     -1 & -1 & 4 & 6  \\
\end{pmatrix}
\]
Observe that $A_1=B_1^T$. 
To relate $A_2$ and $B_2$, observe that applying cut-transpose on $B_2$ with respect to cut $\{b,c\}$, gives us 
\[\begin{pmatrix}
         1 & 2 & -1 & -1 \\ 
     -1 &  2 & -\frac{1}{2} & -\frac{1}{2}   \\
    4 &  4 & 3 & 4   \\
     -2 & -2 & 5 & 6  
\end{pmatrix}.
\]
The obtained matrix is diagonally similar to $A_2$ (they are related by diagonal matrix $D=\diag(1,-\frac{1}{2
},1,1)$).
Hence, the cut-transpose sequence for $A_2$ and $B_2$ is simply $(\{b,c\})$.
To lift this sequence to $A$ and $B$ we have to take union with $\{a,b\}$ (because $\{b,c\}$ contains $b$).
That is, we obtain the sequence $(\{a,b,c\})$. 
Finally, since $A_1=B_1^T$, we need to append this sequence by
another cut $\comp{S}=\{c,d,e\}$.
Hence, the cut-transpose sequence  relating $A$ and $B$
is $(\{a,b,c\}, \{c,d,e\})$. 
Following equation shows this.  
\begin{align*}
A=
\begin{pmatrix}
    1 & 3 & 1 & 1 & 1 \\
    2 & 1 & -1 & -1 & -1 \\ 
    1 & 2 &  2 & 1 & 1   \\
    2 & 4 &  -2 & 3 & 4   \\
    -1 & -2 & 1 & 5 & 6  \\
\end{pmatrix}
&
\xrightarrow[\text{along } \{a,b,c\}]{\text{cut-transpose}}
 %
%  \begin{pmatrix}
%    1 & 3 & 1 & -2 & 1 \\
%     2 & 1 & -1 & 2 & -1 \\ 
%     1 & 2 &  2 & -2 & 1   \\
%     -1 & -2 &  1 & 3 & 5   \\
%     -1 & -2 & 1 & 4 & 6  \\
% \end{pmatrix}\\
 \begin{pmatrix}
   1 & 3 & 1 & 2 & -1 \\
    2 & 1 & -1 & -2 & 1 \\ 
    1 & 2 &  2 & 2 & -1   \\
    1 & 2 &  -1 & 3 & 5   \\
    1 & 2 & -1 & 4 & 6  \\
\end{pmatrix}\\
&
\xrightarrow[\text{along } \{c,d,e\}]{\text{cut-transpose}}
%
% \begin{pmatrix}
%     1 & 2 & 1 & -2 & 1 \\
%     3 & 1 & 2 & -4 & 2 \\ 
%     1 & -1 &  2 & -2 & 1   \\
%     -1 & 1 &  1 & 3 & 5   \\
%     -1 & 1 & 1 & 4 & 6  \\
% \end{pmatrix}
\begin{pmatrix}
    1 & 2 & 1 & 2 & -1 \\
    3 & 1 & 2 & 4 & -2 \\ 
    1 & -1 &  2 & 2 & -1   \\
    1 & -1 &  -1 & 3 & 5   \\
    1 & -1 & -1 & 4 & 6  \\
\end{pmatrix}
=B
\end{align*}
 
\paragraph{An efficient algorithm.}
Now we describe some of the ideas involved in our polynomial time algorithm to find a cut-transpose sequence for two irreducible PME matrices. 
The lemmas mentioned above all have constructive proofs,
that is, the following tasks can be done in polynomial time. 
\begin{itemize}
\item Given a minimal cut $S$ of matrix $A$, we can check whether 
it is also a cut of matrix $B$. If not, then we can find an appropriate cut in $B$ such that applying cut-transpose along it gives us a matrix that has $S$ as a cut. 

\item Given a cut-transpose sequence for $A_2$ and $B_2$ (as defined above), we can find one for $A$ and $B$. 
\end{itemize}
Two parts which remain unclear are -- (i) how to find a minimal cut  of a matrix efficiently and (ii) how to compute $(A+Z)^{\mathrm{adj}}$ efficiently for a given matrix $A$? 

To find a cut of a matrix $A$, we first show that the function 
$f(X) := \rank(A[X,\comp{X}])+\rank(A[\comp{X},X])$
is  a submodular function (Lemma~\ref{lem:MinimalCutInPolynomialTime}).
Then observe that if an irreducible matrix $A$ has a cut,
 then  cuts are precisely those sets $X$
 which minimize $f(X)$ under the constraints $\abs{X} \geq 2$ and $\abs{\comp{X}} \geq 2$.
To find an inclusion-wise minimal cut, 
we simply find a minimum size cut,
using the known algorithms for submodular function minimization
under such constraints (Lemma~\ref{lem:MinimalCutInPolynomialTime}). 

Coming to the second question, recall that 
instead of matrix $A$, we consider  
 $(A+Z)^{\mathrm{adj}}$  to ensure that all matrix entries are nonzero. 
 Here $Z$ is a diagonal matrix with distinct algebraically independent elements (or indeterminates). 
 However, we cannot compute (or even write down) the entries $(A+Z)^{\mathrm{adj}}$ efficiently (succinctly). 
 For efficiency, we need to replace the diagonal entries in $Z$
 with elements from the given field (or a large enough algebraic extension).
 Using ideas from polynomial identity testing, we show that 
 in (deterministic) polynomial time, we can compute an appropriate matrix $Z$, which 
 ensures that the entries of $(A+Z)^{\mathrm{adj}}$ are all nonzero (Claim~\ref{cl:YSubstitution}). 

 \paragraph{Applications to PIT.}
 As discussed earlier, our algorithm to test principal minor equivalence of two matrices $A$ and $B$ can also be viewed as 
 an algorithm to test if the following is a polynomial identity:
 \[\det(A+Y) = \det(B+Y),\]
 where $Y$ is a diagonal matrix with its diagonal entries being all distinct variables. 
 Theorem~\ref{thm:main-three} considers a more general PIT question: 
 whether $\det(A_0+A_1y_1+\ldots+A_my_m)=\det(B_0+B_1y_1+\ldots+B_my_m)$
 for given rank-1 matrices $A_1, A_2, \dots, A_m, B_1, B_2, \dots, B_m$ and arbitrary matrices $A_0, B_0$. 
We get a deterministic algorithm for 
 this more general PIT question via a reduction to testing principal minor equivalence of two given matrices (Section~\ref{sec:PIT}).
 The reduction uses matroid intersection as a subroutine, 
 which is known to be solvable in deterministic polynomial time. 

\section{Notation and Preliminaries}
\label{sec:prelim}

We use $[n]$ to denote the set of positive integers $\{1,2,\ldots, n\}$. 
For any $X\subseteq [n]$, $\comp{X}$ denotes the complement set $X$. 
For two sets $S$ and $T$, $S\Delta T$ denotes the symmetric difference of $S$ and $T$.
For a set $X$ and an element $e$, we use $X+e$ to denote the set $X\cup\{e\}$ and $X-e$ to denote the set $X\setminus \{e\}$.

Suppose that $w_1=(w_{1,1},w_{1,2},\ldots, w_{1,k_1})^T, \ldots, w_\ell=(w_{\ell,1}, w_{\ell,2}, \ldots, w_{\ell,k_\ell})^T$ are $\ell$ vectors over a field $\F$. 
Then, we use $(w_1 \mid \cdots\mid w_\ell)$ to denote the concatenation of the vectors $w_1,\ldots, w_\ell$ as follows $$(w_1 \mid \cdots \mid w_\ell)=(w_{1,1},\ldots, w_{1,k_1},\ldots w_{\ell,1},\ldots, w_{\ell, k_\ell})^T.$$ 

For an $n\times n$ matrix $A$ and $S, T\subseteq [n]$, $A[S, T]$ denotes the submatrix of $A$ with rows indexed by elements in $S$ and columns indexed by elements in $T$.
For $S\subseteq [n]$, let $A[S]$ denote the submatrix $A[S,S]$.
When $S=\{i\}$, then $A[i,T]=A[S,T]$. 
We follow a similar notation when $T$ is a singleton. 
%For $S=T$, we use $A(S)$ to denote the submatrix $A[S, S]$.
For a square matrix $A$, by $A^{\mathrm{adj}}$, we denote the adjoint, or adjugate, of $A$.

\subsection{Principal minor equivalence} 
Suppose that $A$ and $B$ are two $n\times n$ matrices over any field. 
The matrix $A$ is said to be \emph{principal minor equivalent} to $B$ if the corresponding principal minors of $A$ and $B$ are \emph{equal}, i.e. for all $S\subseteq [n]$, $\det(A[S, S])=\det(B[S, S])$. 
We use $A\PME B$ to denote that $A$ is the principal minor equivalent to $B$. 

The following lemma shows that the principal minor equivalence relation between two matrices remains unchanged under adjoint operation and shift by appropriate diagonal matrices. 
It is a straightforward consequence of~\cite[Lemma~4]{Hartfiell84}.
\begin{lemma} 
\label{lem:inversibleAdjugate}
Let $A$ and $B$ be two $n\times n$ matrices over a field $\F$.
Let $D$ be an $n\times n$ diagonal matrix over $\F$ such that $A+D$ and $B+D$ are non-singular. Then, $A\PME B$ if and only if $\adj{A+D}\PME \adj{B+D}$.
\end{lemma}
 % \begin{proof}
 %     From \cref{lem:AdjugateNonZero}, $\adj{A+Y}\PME \adj{B+Y}$. Note that entries of $\adj{A+Y}$ and $\adj{B+Y}$ are from $\F[y_1,\dots,y_n]$. Hence, we can put $Y=D$ in $\adj{A+Y}$ and $\adj{B+Y}$. This implies $\adj{A+D}\PME \adj{B+D}$. 

 %     Now suppose, $\adj{A+D}\PME \adj{B+D}$. Since $A+D$ is invertible, using Jacobi's identity, for any subset $S\subset [n],$ $\det((A+D)[S])=\frac{\det(\adj{A+D}[\comp{S}])}{\det(A+D)^{|\comp{S}|-1}}$. Similarly, $\det((B+D)[S])=\frac{\det(\adj{B+D}[\comp{S}])}{\det(B+D)^{|\comp{S}|-1}}$. Hence, $A+D\PME B+D$. Therefore, from \cref{cor:AddDiagonal}, \[(A+D)-D=A\PME (B+D)-D=B.\] 
 % \end{proof}

\subsection{Reducible and Irreducible matrix}
\begin{definition}[Reducible and Irreducible matrix] 
\label{def:Reducible}
A matrix is called \emph{reducible} if it can be written as a block upper triangular matrix after permuting the rows and the corresponding columns. A matrix that is not reducible is called \emph{irreducible}. 

Equivalently, if we replace the nonzero off-diagonal entries with one and the diagonal entries with zero, then a reducible matrix corresponds to the adjacency matrix of a directed graph having more than one strongly connected component.
\end{definition}
From the above definition, it is easy to see that any matrix $A$ with all nonzero off-diagonal entries is an irreducible matrix. The above definition directly gives us the following observation.
\begin{observation}
\label{ob:irreducible-matrix-and-directed-graph}
Let $A$ be an $n\times n$ matrix over a field $\F$ such that the row and columns of $A$ are indexed by $[n]$. Let $G_A$ be a directed graph defined as follows: the vertex set in $[n]$, and a tuple $(i,j)$ is an edge of $G_A$ if and only if $i\neq j$ and $A[i,j]\neq 0$. Let $I_1, I_2,\ldots, I_s$ be the strongly connected components of $A$. Then, after permuting the rows and the corresponding columns, the matrix $A$ can be made a block upper triangular matrix, and the diagonal blocks $A(I_1), A(I_2),\ldots, A(I_s)$ are irreducible matrices. 
\end{observation}

For two reducible matrices $A$ and $B$, the next lemma helps to reduce the testing of whether $A\PME B$ to multiple instances of testing whether two irreducible matrices have the same corresponding principal minors. The following lemma is a direct consequence of~\cite[Corollary~5.4]{Ahmadieh23}.
\begin{lemma}
\label{lem:redToIr}
Let $A$ and $B$ two $n\times n$ matrices over  a field $\F$. Suppose that after permuting the rows and the corresponding columns, $A$ can be written as a block upper triangular matrix with $s$ diagonal blocks $A_1,A_2,\dots,A_s$ where each $A_i$ is irreducible and the rows and columns of $A_i$ are indexed by set $T_i\subseteq [n]$. Then, $A\PME B$ if and only if the following holds.
\begin{enumerate}
    \item After permuting some rows and the corresponding columns, $B$ can be written as a block upper triangular matrix with $s$ diagonal blocks $B_1,B_2,\dots,B_s$ such that each $B_i$ is irreducible and the rows and columns of $B_i$ are indexed by set $T_i$.
    \item For each $i\in [s]$, $A_i\PME B_i$.
\end{enumerate}
\end{lemma}

\subsection{Cut of a matrix}
\begin{definition}[Cut of a matrix]
\label{def:cut}
Let $A$ be an $n\times n$  matrix over a field $\F$ such that $n\geq 4$. A subset $X\subset [n]$ is called a \emph{cut} in $A$ if $2\leq |X| \leq n-2$ and the rank of the submatrices $A[X, \comp{X}] \text{ and } A[\comp{X}, X]$ are at most one.

In particular, if $A$ is an irreducible matrix and $X$ is cut in $A$, then $\rank(A[X,\comp{X}])=\rank(A[\comp{X},X])=1$.
\end{definition}
By definition, if $X$ is a cut of a matrix $A$ then so is $\comp{X}$.
For a matrix $A$, a cut $X$ in $A$ is called a \emph{minimal cut} if there exists no other cut $X'$ in $A$ such that $X'\subseteq X$. Note that any cut of size two is always a minimal cut.
% Rohit: Obs ob:cut-deciding-poly-time can be removed.

Next, we show that the set of cuts of a matrix remains the same if we take its adjugate after adding an appropriate diagonal matrix. 
\begin{lemma}
\label{lem:AdjugateCut}
Let $A$ be an $n\times n$ matrix over a field $\F$. Let $D$ be an $n\times n$ diagonal matrix over $\F$ such that $A+D$ is non-singular. Then, $A$ and $\adj{A+D}$ have the same set of cuts. 
\end{lemma}
For proof, see~\cref{appendix:missing-proofs-from-prprelim}.

\subsection{Diagonal similarity}
\label{sec:diag-equiv}

Suppose that $A$ and $B$ are two $n\times n$ matrices over a field $\F$. We say that $A$ is \emph{diagonally similar} to $B$, denoted by $A\DE B$, if there exists an $n\times n$ invertible diagonal matrix $D$ over $\F$ such that $B=DAD^{-1}$. We say that $A$ and $B$ are \emph{diagonally equivalent}, denoted by $A\dge B$, if $A\DE B$ or $A\DE B^T$.

In the following claim, we describe how to efficiently check whether two matrices are diagonally similar or not.
\begin{claim}
\label{cl:DE-in-poly-time}
Given two $n\times n$ matrices $A$ and $B$ over $\F$, in polynomial time, we can decide whether $A\DE B$.
\end{claim}
\begin{proof}[Proof Sketch]
%Suppose $A\DE B$, i.e., there is an invertible diagonal matrix $D$ over $\F$ such that $B=DAD^{-1}$.
%
%Then for any $i, j$ such that $A[i,j]\neq 0$, it must be that $B[i,j]\neq 0$ and $B[i,j]/A[i,j] = D[i] / D[j]$. 
%
Observe that if $A \DE B$ then we must have 
an invertible diagonal matrix $D$ such that 
$B[i,j]/A[i,j] = D[i] / D[j]$, for any $i\neq j$ with $A[i,j]\neq 0$.
%
%More generally, for any sequence of indices $(i_1, i_2, \dots, i_k)$,
%we must have 
%$$ B[i_1,i_2]B[i_2,i_3] \cdots B[i_{k-1},i_k] B[i_k, i_1]= 
%A[i_1,i_2]A[i_2,i_3] \cdots A[i_{k-1},i_k] A[i_k,i_1].$$ 
%
Consider a weighted directed graph $G$ on $n$ vertices such that
$(i,j)$ is an edge for $i\neq j$ if and only if 
$ A[i,j]\neq 0$ or $A[j,i] \neq 0 $. 
Let us define the weight of an edge $(i,j)$ as 
$w(i,j) = B[j,i]/A[j,i]$ or $A[i,j]/B[i,j] $ whichever is defined. 
If both are defined, they must be equal, otherwise $A$ and $B$
cannot be diagonally similar.
Observe that for any path $(i_0, i_1, \dots, i_k)$ in graph $G$,
it must be that
\[ D[i_k]/D[i_0] = w(i_0, i_1) w(i_1, i_2) \cdots w(i_{k-1}, i_k) \]
Moreover, any diagonal matrix satisfying the above equation for all paths in $G$ will give us the desired diagonal matrix $D$.
So, we construct $D$ for given matrices $A$ and $B$ as follows.
For any connected component in $G$, 
pick an arbitrary vertex $i$ from the component and set $D[i]=1$. 
For any other vertex $j$ in that component, find a path
$(i=i_0,i_1, i_2, \dots, i_k= j)$ and set 
$D[j]= w(i_0, i_1) w(i_1, i_2) \cdots w(i_{k-1}, i_k) .$
Repeating this for every component in $G$ will give us matrix $D$.
Finally we should check that $B[i,j]/A[i,j] = D[i] / D[j]$, for every $i\neq j$ with $A[i,j]\neq 0$.

%Let $V$ be the vector space containing all diagonal matrices $D$ such that $BD=DA$. Using the polynomial time algorithm for solving homogeneous linear systems, we can find a basis $D_1, D_2,\ldots, D_k$ for $V$. Observe that $A\DE B$ if and only if there exists an invertible diagonal matrix $D\in V$. Let $\widetilde D$ be a diagonal matrix defined as $\widetilde D=\sum_{i=1}^ny^{i-1}D_i$. The vector space $V$ contains an invertible diagonal matrix if and only if all the diagonal entries of $\widetilde D$ are nonzero polynomials in $y$. Next, if all the diagonal entries of $\widetilde D$ are nonzero, find an element $\alpha \in\F$ such that the evaluation of $\prod_{i=1}^n\widetilde D[i,i]$ at $\alpha$ is nonzero.   
\end{proof}

One can observe that if $A\dge B$, then $A\PME B$.
Next, we consider the converse direction. Hartfiel and Loewy~\cite[Theorem~3]{Hartfiell84} showed that when $n=2$ or $3$, and $A$ is an irreducible matrix, $A\PME B$ implies that $A\dge B$ . 
Later, Lowey~\cite[Theorem~1]{Loewy86} showed that if $A$ is an irreducible matrix and has no cut, then $A\PME B$ implies $A\dge B$. Therefore, by combining them, we have the following lemma:

\begin{lemma} 
\label{lem:noCut}
Let $A$ and $B$ two $n\times n$ matrices over a field $\F$ such that $A$ is irreducible and $A\PME B$. Then, the following holds:
\begin{enumerate}
\item if $n=2$ or $3$, then $A\dge B$.
\item if $n\geq 4$ and $A$ has no cut, then $A\dge B$.
\end{enumerate}
\end{lemma}

The next lemma shows that the diagonal similarity relation also holds for the adjugate of the matrices obtained by adding an appropriate diagonal matrix. It directly follows from~\cite[Lemma~4]{Hartfiell84}.
\begin{lemma}
\label{lem:DE-through-adjoint}
Let $A$ and $B$ be two $n\times n$ matrices over $\F$. Let $D$ be an $n\times n$ diagonal matrix such that both $A+D$ and $B+D$ are invertible. Then, $A\DE B$ if and only if $\adj{A+D}\DE \adj{B+D}$.
\end{lemma}

\subsection{Cut-transpose operation}
\label{subsec:twist-operation}
In the previous section, we have seen that under diagonal similarity, the values of the principal minors of a matrix remain unchanged.
Now, we describe another operation under which also the values of the principal minors remain the same. 
This operation was defined by Ahmadieh~\cite[Lemma~4.5]{Ahmadieh23},
and we refer to it as cut-transpose.
%The definition of this operation is inspired by~\cite[Lemma~4.5]{Ahmadieh23}.

\begin{definition}[Cut-transpose operation]
\label{def:twist-operation}
Let $A$ be an $n\times n$ irreducible matrix represented as follows, and $X\subseteq [n]$ be a cut of $A$. Let  $q, u\in\F^{|\comp{X}|}$ such that $q$ is the first non-zero row of $A[X,\comp{X}]$, $u$ is the first non-zero column of $A[\comp{X},X]$. Let $p,v\in\F^{|X|}$  such that $A[X,\comp{X}]=p\cdot q^T$ and $A[\comp{X},X]=u\cdot v^T$. \[ 
A=\begin{pmatrix}
A[X]        & p\cdot q^T\\
&\\
u\cdot v^T  & A[\comp{X}]
\end{pmatrix}.
\] 
 Then, the \emph{cut-transpose operation on $A$ with respect to $X$}, denoted by $\tw(A,X)$, transforms $A$ to the following matrix:
\[
\tw(A,X) =
\begin{pmatrix}
A[X]          & p\cdot u^T\\
&\\
q\cdot v^T  & A[\comp{X}]^T
\end{pmatrix}.
 \] 
% % Then, the \emph{twist operation} on $A$ with respect to $X$ produces a family of matrices, denoted by $\tw(A, X)$, defined as follows: 
% % we define the \emph{twist} operation on matrix $A$ with respect to cut $X$ denoted by $\tw(A,X)$ as follows
% If the matrix $A$ is written as 
% \[ 
% \begin{pmatrix}
% A(X)        & p\cdot q^T\\
% &\\
% u\cdot v^T  & A(\comp{X})
% \end{pmatrix},
% \] 
% where $p,v\in\F^{|X|}$ and $q, u\in\F^{|\comp{X}|}$, then the following matrix 
% % Then, the \emph{twist operation} on $A$ with respect to $X$ transforms $A$ to a new matrix $\tw(A,X)$ as follows:
% \[
% \begin{pmatrix}
% M           & p\cdot u^T\\
% &\\
% q\cdot v^T  & N^T
%  \end{pmatrix}
%  \] 
% is in $\tw(A, X)$.

% In particular, if $|X|=n-1$, then $A\in \tw(A,X)$, and if $|X|=1$, $A^T\in\tw(A,X)$. 
% \[\:\: \tw(A,X) \text{ denote the matrix }\begin{pmatrix}
%             M & p\cdot u^T\\
%             q\cdot v^T & N^T
%     \end{pmatrix}\] 
\end{definition}
\begin{remark}
\label{remark:twist}
Note that in the above definition, we take one particular rank-one decomposition for submatrices $A[X,\comp{X}]$ and $A[\comp{X},X]$.
For every nonzero $\alpha, \beta\in \F$, the rank-one submatrices $A[X, \comp{X}]$ and $A[\comp{X}, X]$ are  equal to $(\alpha p)\cdot (q/\alpha)^T$ and $(\beta u) \cdot (v/\beta)^T$, respectively. Depending on what rank one decomposition we choose, we can get a different matrix after applying this operation. However, all these obtained matrices are \emph{diagonally similar} to each other. Also, one advantage of choosing the above decomposition is that $\tw(\tw(A,X),X)=A.$
\end{remark}

% \begin{remark}
% \label{remark:twist}
% Note that in the above definition, for every nonzero $\alpha, \beta\in \F$, the rank-one submatrices $A[X, \comp{X}]$ and $A[\comp{X}, X]$ are  equal to $(\alpha p)\cdot (q/\alpha)^T$ and $(\beta u) \cdot (v/\beta)^T$, respectively. Depending on what rank one decomposition we choose, we can get a different matrix after applying the twist operation, and let $\tw(A, X)$ be exactly the set of all such matrices. However, all the matrices in $\tw(A, X)$ are \emph{diagonally equivalent} to each other. Therefore, all the matrices in $\tw(A, X)$ have the same corresponding principal minors. Thus, slightly abusing the notation, we also use $\tw(A, X)$ to denote any matrix we can get after applying the twist operation on $A$ with respect to $X$.
% \end{remark}

For a $k\times \ell$ matrix $M$ with $\rank(M)\leq 1$, in polynomial time, we can find $p\in\F^k$ and $q\in\F^\ell$ such $M=p\cdot q^T$. Hence, we can find the cut-transpose of a matrix with respect to a given cut in polynomial time.

Now, we mention some properties of the cut-transpose operation. First, we show that under cut-transpose operation, the values of the principal minors of a matrix remain the same.
\begin{lemma} 
\label{lem:PMEwithTwist}
Let $A$ be an $n\times n$ irreducible matrix over a field $\F$ with a cut $X\subset [n]$. Then, $A\PME \tw(A,X)$.
\end{lemma}
% \sgc{May need to fix based on twist operation definition change.}
For proof, see~\cref{appendix:missing-proofs-from-prprelim}. Next, we show that the cut-transpose operation and the adjoint operation commute with each other up to diagonal similarity.
% Next, we show that the cut-transpose operation on a matrix $A$ is carried over to the matrix $\adj{A+Y}$. Formally, we show the following lemma.

\begin{lemma} 
\label{lem:AdjugateTwistEquality}
Let $A$ be an $n\times n$ irreducible matrix over a field $\F$. Then, a cut $X\subseteq [n]$ of $A$ is also a cut of $A^{\mathrm{adj}}$ and $$\tw(A, X)^{\mathrm{adj}}\DE \tw(A^{\mathrm{adj}}, X).$$ 
\end{lemma}
%%%
% Rohit: below explanation can be removed.
%%%%%%
% From~\cref{remark:twist}, $\tw(A, X)$ is a family of matrices such that all the matrices are diagonally similar. The set $\tw(A, X)^{\mathrm{adj}}$ consists of all the matrices we get after taking the adjoint of the matrices in $\tw(A, X)$. For nonzero $\alpha, \beta\in\F$, let $D_{\alpha,\beta}$ be an $n\times n$ diagonal matrix defined as follows: $D_{\alpha, \beta}[i,i]=\alpha$ for $i\in X $ and $\beta$ otherwise. Then, from~\cref{remark:twist} and~\cite[Lemma~4]{Hartfiell84}, if $\widetilde A\in \tw(A,X)$, $$\tw(A,X)^{\mathrm{adj}}=\{D_{\alpha,\beta}\cdot \widetilde A^{\mathrm{adj}}\cdot D_{\alpha, \beta}^{-1}\, \mid\, \alpha, \beta\in\F\setminus\{0\}\}.$$ On the other hand, from~\cref{cl:submatrix-in-adj}, both $\rank(A^{\mathrm{adj}}[X, \comp{X}])$ and $\rank(A^{\mathrm{adj}}[\comp{X}, X])$ are at most one. Thus, if $A'\in \tw(A^{\mathrm{adj}}, X)$, $$\tw(A^{\mathrm{adj}}, X)=\{D_{\alpha,\beta }\cdot A'\cdot D_{\alpha, \beta}^{-1}\,\mid\, \alpha, \beta\in\F\setminus\{0\}\}.$$ In the proof, we show that there exists a common matrix in $\tw(A, X)^{\mathrm{adj}}$ and $\tw(A^{\mathrm{adj}}, X)$, thus implying the above lemma. 

For proof, see~\cref{appendix:missing-proofs-from-prprelim}. Next, we define the \emph{cut-transpose equivalence} relation. \cref{thm:main-one} says that it characterizes principal minor equivalence for irreducible matrices.

\begin{definition}
\label{def:property-P}
Let $A$ and $B$ be two irreducible matrices over a field $\F$ with their rows and columns indexed by $I$. Let $\mathcal X= (X_1,X_2,\ldots,X_k)$ be a sequence of subsets of $I$.  We say that $A$ and $B$ are \emph{cut-transpose equivalent} with respect to \emph{cut sequence} $\mathcal X$ if it  produces a sequence of matrices $(A_0=A, A_1, A_2, \ldots, A_k)$ with the following property: $$\forall i\in[k],\ A_i=\tw(A_{i-1}, X_i) \text{ where } X_i \text{ is a cut in } A_{i-1},\text{ and }A_k\dge B.$$
\end{definition}
The following lemma demonstrates how the cut-transpose relation extends to the adjoint.
\begin{lemma}
    \label{lem:ResultForAplusDimpliesForA} Let $A,B$ be two $n\times n$ irreducible matrices over a field $\mathbb{F}$ and $\mathcal{X}$ be sequence of subsets of $[n]$. Let $D$ be a diagonal matrix such that $A+D$ and $B+D$ are invertible. Then, $A$ and $B$ are cut-transpose equivalent with respect to $\mathcal{X}$ if and only if $\adj{A+D}$ and $\adj{B+D}$ are cut-transpose equivalent with respect to  $\mathcal{X}$. 
\end{lemma}
\begin{proof}
    We show this by induction on the size of $\mathcal{X}$. For the base case, $\mathcal{X}$ is empty
    %Note that $A\DE B$ if and only if $A+D\DE B+D$. 
    and 
    from \cref{lem:DE-through-adjoint}, $A\DE B$ if and only if $\adj{A+D}\DE \adj{B+D}$. Similarly, $A\DE B^T$ if and only if $\adj{A+D}\DE \adj{B^T+D}=(\adj{B+D})^T$. 

    For the inductive step, we assume that the statement is true when the sequence length is less than $k$. Now, we show this for $\mathcal{X}=(X_1,X_2,\dots, X_k)$. Here, we show only the forward direction. The other direction can be shown similarly. Let $(A'_0=\adj{A+D},A'_1,\dots,A'_k)$ be the sequence of matrices such that for each $i\in [k]$, $X_i$ is cut of $A'_{i-1}$  and $A'_i=\tw(A'_{i-1},X_i)$. Since $A'_k \dge \adj{B+D}$, $X_k$ is also a cut of $\adj{B+D}$. From \cref{lem:AdjugateTwistEquality}, $\tw(\adj{B+D},X_k)\DE\adj{\tw(B+D,X_k)}.$ 
    
    Since $A'_k\dge \adj{B+D}$, $A'_k$ is diagonally similar to either $\adj{B+D}$ or its transpose. Suppose $A'_k\DE \adj{B+D}$. Then, $\tw(A'_k,X_k)\DE \adj{\tw(B+D,X_k)}.$ Note that $X_k$ is also a cut of $A'_k$ and $\tw(A'_k,X_k)=A'_{k-1}$, as 
    cut-transpose operation is its own inverse. 
    %applying cut-transpose operation twice with respect to the same cut gives back the original matrix. 
    Hence, $A'_{k-1}\DE \adj{\tw(B+D,X_k)}=\adj{\tw(B,X_k)+D}.$ This implies $\adj{A+D}$ and $\adj{\tw(B,X_k)+D}$ are cut-transpose equivalent with respect to sequence $\mathcal{X}-X_k$. By induction hypothesis, $A$ and $\tw(B,X_k)$ are cut-transpose equivalent with respect to $\mathcal{X}-X_k$. Since $\tw(\tw(B,X_k),X_k)=B$, we get that $A$ and $B$ are cut-transpose equivalent with respect to cut sequence $\mathcal{X}.$ Similarly, we can show for the other case when $A'_k\DE (\adj{B+D})^T$.

\end{proof}

\section{Characterization of Prinicipal Minor Equivalence for Irreducible Matrices}
\label{sec:proof-of-thm-one}

In this section, we show that two irreducible matrices are principal minor equivalent if and only if they are cut-transpose equivalent. Formally, we show \cref{thm:main-one}. The forward direction directly follows from \cref{lem:PMEwithTwist} and the fact that diagonally equivalent matrices are principal minor equivalent. For the other direction, first, we argue that we only need to show the theorem for matrices whose all off-diagonal entries are non-zero. 

Suppose there exists a diagonal matrix $D$ such that $A+D$ and $B+D$ are invertible and off-diagonal entries of $\adj{A+D}$ and $\adj{B+D}$ are non-zero. Since $A\PME B$, from \cref{lem:inversibleAdjugate}, $\adj{A+D}\PME \adj{B+D}$. From \cref{lem:ResultForAplusDimpliesForA}, if $\adj{A+D}$ and $\adj{B+D}$ are cut-transpose equivalent with respect to a cut sequence $\mathcal X$ of size at most $2n$ then so are $A$ and $B$. This implies that if the \cref{thm:main-one} holds for matrices with non-zero off-diagonal entries, then it also holds for general irreducible matrices. We show the existence of such $D$ in \cref{cl:YSubstitution}. Hence, in the rest of this section, we assume that $A$ and $B$ have non-zero off-diagonal entries, without loss of generality.

If $n\leq 3$ or $A$ does not have any cut, then \cref{thm:main-one} directly follows from \cref{lem:noCut}. So, we assume that $n\geq 4$ and $A$ has a cut. We prove the theorem using induction on $n$. The base case of $n=4$ directly follows from the following lemma. See \cref{appendix:otherProofs} for the proof.

\begin{lemma} 
\label{lem:size4Matrix}
Let $A$ be a $4\times 4$ matrix over $\F$ with all off-diagonal entries are nonzero. Let $B$ be another $4\times 4$ matrix over $\F$ such that $A\PME B$. Then, one of the following two holds:
\begin{enumerate}
\item $A\dge B$.
\item There exists a common cut in $A$ and $B$. Furthermore, for any common cut $X$ of $A$ and $B$, $\tw(A,X)\dge B$.
\end{enumerate}
\end{lemma}

Now, we have two $n\times n$  matrices $A$ and $B$ with non-zero off-diagonal entries and at least one cut such that $B\PME A$. First, we show some relation between minimal cuts of $A$ and $B$ that enables us to apply induction. Precisely, we show that if
$A$ and $B$ have the same principal minors, and $A$ has a cut, then a minimal cut of $A$ is also a
cut of $B$ if its size is greater than two. Otherwise, if the size of a minimal cut $S$ of $A$ is two, then either it is
also a cut of $B$, or there exists a cut $X$ in $B$ such that the cut-transpose of $B$ with respect to $X$ has cut $S$. To show this relationship between cuts, we first show the following three results. 

The following lemma establishes the relation between cuts of a matrix and its cut-transpose. 
\begin{lemma} 
\label{thm:cutmap}
Let $A$ be an $n\times n$ matrix over $\F$ with nonzero off-diagonal entries. Let $S\subseteq [n]$ be a cut in $A$. Then, for any $T\subseteq [n]$ the following holds.
\begin{enumerate}
\item If $T$ or $\comp{T}$ is a subset of $S$ or $\comp{S}$, then $T$ is a cut in $A$ if and only if $T$ is a cut in $\tw(A, S).$
\item Otherwise, $T$ is a cut in $A$ if and only if $T\Delta S$ is a cut in $\tw(A,S)$.
\end{enumerate}
\end{lemma}
\begin{proof}
We start with the proof of the first part of the lemma.
\paragraph*{Proof of the first part.} Assume that $T\subseteq S$ and $T$ is a cut in $A$. Then, the matrix $A$ has the following structure: 
\[
A=
\begin{blockarray}{cccc}
&              T                   & S\setminus T         & \comp{S} \\
\begin{block}{c(ccc)}
{T}            & *                 & u_1 \cdot v_1^T      & u_1 \cdot v_2^T\\
&&&\\
{S\setminus T} & p_1 \cdot q_1^T   & *                    & u_2  \cdot v_2^T\\
&&&\\
{\comp{S}}     & p_2\cdot q_1^T    & p_2\cdot q_2^T       & *\\            
\end{block}
\end{blockarray}
\] such that  
$$u_1, q_1\in \F^{|T|},\ \ \ \ \ \ v_1, u_2, p_1, q_2\in \F^{|S|-|T|},\ \   \text{ and }\ \ v_2, p_2\in\F^{|\comp{S}|},$$ and `$*$' marked submatrices can be arbitrary. 
% After applying the twist operation on $A$ with respect to the cut $S$, we get the matrix $\tw(A,S)$ as follows:
% \[
% \tw(A, S)=
% \begin{blockarray}{cccc}
%                &    T              & S\setminus T     & \comp{S}\\
% \begin{block}{c(ccc)}
% {T}            &   A_1             & u_1 \cdot v_1^T  &  u_1\cdot p_2^T\\
% &&&\\
% {S\setminus T} & p_1 \cdot q_1^T   & A_2              & u_2  \cdot p_2^T\\
% &&&\\
% {\comp{S}}     & v_2\cdot q_1^T    & v_2\cdot q_2^T   & A_3^T\\ 
% \end{block}
% \end{blockarray}
% \]
After applying the cut-transpose operation on $A$ with respect to the cut $S$, using \cref{remark:twist}, $$\tw(A,S)[T,\comp{T}]\DE u_1\cdot (v_1 \mid p_2)^T\text{ and } \tw(A,S)[\comp{T}, T]\DE(p_1 \mid v_2)\cdot q_1^T.$$ Therefore, $T$ is also a cut in $\tw(A,S)$. The converse follows because $\tw(\tw(A,S),S) = A$. 

%For the converse direction, observe that $S$ is a cut in $\tw(A, S)$, and $A$ can be seen as a matrix we get after applying the cut-transpose operation on $\tw(A, S)$ with respect to $S$. Therefore, the above analysis also says that $T$ will be a cut in $A$ if it is a cut in  $\tw(A, S)$ and $T\subseteq S$.

Now we assume that $T\subseteq \comp{S}$. Note that the set of cuts in $A$ is the same as the set of cuts in $A^T$. 
% Therefore, 
% \begin{enumerate}
% \item $S$ is also a cut in $A^T$. 
% \item $T$ is a cut in $A$ if and only if it is a cut in $A^T$. 
% \end{enumerate} 
Since $T\subseteq \comp{S}$, from the above discussion, $T$ is a cut in $A^T$ if and only if $T$ is a cut in $\tw(A^T,\comp{S})$. Observe that $\tw(A^T,\comp{S})\DE\tw(A,S)$. Thus, when $T\subseteq \comp{S}$, the set $T$ is a cut in $A$ if and only if $T$ is a cut in $\tw(A,S)$. The proof for the remaining cases directly follows from these.

\paragraph*{Proof of the second part.} Assume that neither $T$ nor $\comp{T}$ is  a subset of $S$ or $\comp{S}$, and $T$ is a cut in $A$. This implies that $S\setminus T$, $S\cap T$ and $T\setminus S$ are nonempty. Since $S$ is a cut in $A$, the matrix $A$ has the following structure.
\begin{equation}
\label{eqn:str-A-wrt-cut-S}
A=
\begin{blockarray}{ccccc}
                  & S\setminus T    & S\cap T         & T\setminus S    & \comp{S \cup T}\\
\begin{block}{c(cccc)}
{S\setminus T}    &   *             &  *              & u_1\cdot v_1^T  &  u_1\cdot v_2^T\\
&&&&\\
{S\cap T}         &   *             & *               & u_2\cdot v_1^T  &  u_2\cdot v_2^T\\
&&&&\\
{T\setminus S}    & p_1\cdot q_1^T  & p_1\cdot q_2^T  & *               &  *\\
&&&&\\
{\comp{S \cup T}} & p_2\cdot q_1^T  & p_2\cdot q_2^T  & *               &  *\\ 
\end{block}
\end{blockarray}
\end{equation}
such that $$u_1, q_1\in\F^{|S\setminus T|},\ \ \ \  v_1, p_1\in \F^{|T\setminus S|},\ \ \ \  v_2, p_2\in \F^{|\comp{S\cup T}|},\  \text{ and }\   u_2, q_2\in\F^{|S\cap T|}.$$
Since $T$ is also a cut, the columns and rows of $A[S\setminus T, S\cap T]$ are multiples of $u_1$ and $q_2^T$ respectively. Hence, $A[S\setminus T, S\cap T]=(\alpha u_1)\cdot q_2^T $ for some $\alpha \neq 0$. Since $\rank(A[\comp{T},T])=1,$ $A[\comp{S\cup T},T\setminus S]=p_2\cdot (v_1^T/\alpha)$. Similarly, for some non-zero $\beta$, $$A[S\cap T,S\setminus T]= (\beta u_2)\cdot q_1^T \text{ and } A[T\setminus S,\comp{S\cup T}]=p_1\cdot (v_2^T/\beta).$$
Thus, the matrix $A$ has the following form.
\begin{equation}
\label{eqn:str-A}
A=
\begin{blockarray}{ccccc}
                  & S\setminus T                  & S\cap T                 & T\setminus S             &\comp{S \cup T}\\
\begin{block}{c(cccc)}
{S\setminus T}    & *                             &  (\alpha u_1)\cdot q_2^T & u_1\cdot v_1^T           &u_1\cdot v_2^T\\
&&&&\\
{S\cap T}         &  (\beta u_2)\cdot q_1^T   & *                       & u_2\cdot v_1^T           &u_2\cdot v_2^T\\
&&&&\\
{T\setminus S}    & p_1\cdot q_1^T                & p_1\cdot q_2^T          & *                        &p_1\cdot (v_2^T/\beta)\\
&&&&\\
{\comp{S \cup T}} & p_2\cdot q_1^T                & p_2\cdot q_2^T          & p_2\cdot (v_1^T/\alpha)&  *\\ 
\end{block}
\end{blockarray}
\end{equation}
From~\cref{eqn:str-A}, applying cut-transpose operation on $A$ with respect to the cut $S$, we get that 
% we apply cut-transpose operation on $A$ with respect to the cut $S$ and get 
% \begin{equation}
% \label{eqn:str-A-twist}
% \tw(A,S)=
% \begin{blockarray}{ccccc}
%                   & S\setminus T                 & S\cap T                 & T\setminus S             & \comp{S \cup T}\\
% \begin{block}{c(cccc)}
% {S\setminus T}    & A_1                          &  \zeta(u_1\cdot q_2^T)  & u_1\cdot p_1^T           &  u_1\cdot p_2^T\\
% &&&&\\
% {S\cap T}         & \omega^{-1}(u_2\cdot q_1^T)  & A_2                     & u_2\cdot p_1^T           &  u_2\cdot p_2^T\\
% &&&&\\
% {T\setminus S}    & v_1\cdot q_1^T               & v_1\cdot q_2^T          & A_3^T                    &  \zeta^{-1}(v_1\cdot p_2^T)\\
% &&&&\\
% {\comp{S \cup T}} & v_2\cdot q_1^T               & v_2\cdot q_2^T          &  \omega(v_2\cdot p_1^T)       &  A_4^T\\ 
% \end{block}
% \end{blockarray}
% \end{equation}
% In~\cref{eqn:str-A-twist},
\begin{eqnarray*}
\tw(A,S)[T\Delta S,\comp{T\Delta S}]   &\DE&   (\alpha u_1 \mid v_1)\cdot (q_2\mid \alpha ^{-1}p_2)^T, \text{ and }\\ 
\tw(A,S)[\comp{T\Delta S}, T\Delta S]  &\DE&   (\beta u_2 \mid v_2)\cdot (q_1 \mid \beta^{-1} p_1)^T.
\end{eqnarray*}
Therefore, $S\Delta T$ is a cut in $\tw(A,S)$. 

Converse follows, because $\tw(\tw(A,S),S) = A$ and $(T\Delta S)\Delta S=T$.
%For the converse direction, assume that $T\Delta S$ is a cut in $\tw(A,S)$. As mentioned earlier, $A$ can be see as the matrix we get after applying cut-transpose operation on $\tw(A, S)$ with respect to the cut $S$. Therefore, the above discussion implies that if $T\Delta S$ is a cut in $\tw(A,S)$, then $(T\Delta S)\Delta S=T$ is a cut in $A$. This completes the proof of our lemma.
\end{proof}
In the following lemma, we state a property about a minimal cut  of size greater than two.
\begin{lemma} \label{lem:minCutGeneral}
Let $A$ be an $n\times n$ matrix over $\F$ such that the off-diagonal entries of $A$ are nonzero. Let $S$ be a minimal cut in $A$ of size greater than two. Let $T$ be a nonempty subset of $\comp{S}$, $X\subseteq S\cup T$ and $\widetilde X=(S\cup T)\setminus X$. Then, if $X$ is a cut in $A[S\cup T]$, then either $S\subseteq X$ or $S \subseteq \widetilde X$. 

In particular, if $T=\{t\}$ for some $t\in\comp{S}$, then the matrix $A[S+t]$ have no cut.
\end{lemma}
\begin{proof}
For the sake of contradiction, assume that $X\subseteq S\cup T$ is a cut of $A[S \cup T]$ such that neither $S\subseteq X$ nor $S \subseteq \widetilde X$. Since $|S|\geq 3$, either $|S\cap X|\geq 2$ or $|S\cap \widetilde X|\geq 2$. This implies that a cut exists in $A[S \cup T]$, which contains at least two elements of $S$. Hence, without loss of generality, we can assume that $|S\cap X|\geq 2$. If $T\setminus X$ is empty, then $\widetilde X \subseteq S$ and hence $|\widetilde X \cap S|= |\widetilde X|\geq 2$. Also, $T\setminus \widetilde X = T$ is non-empty. This implies that $A[S \cup T]$ has a cut such that it has at least two elements from $S$, and $T$ has at least one element that is not present in it. Without loss of generality, we assume that $|X\cap S|\geq 2$ and $T\setminus X\neq \emptyset$. Since $S$ is not a subset of $X$, $S\setminus X$ is non-empty.

Let $t\in T\setminus X= T\cap \widetilde X$. Let $$A[S\cap X,t]=u_1, A[S\setminus X,t]=u_2, A[t,S\cap X]=q_1^T \text{ and } A[t, S\setminus X]= q_2^T.$$ Since $X$ is a cut of $A[S+T]$ and $t\in \widetilde X$, the columns of $A[S\cap X, S\setminus X]$ are multiples of $u_1$. Similarly, the rows of $A[S\setminus X, S\cap X]$ are multiples of $q_1^T$.
Since $S$ is a cut of $A$, $A$ has the following structure.

\begin{equation}
\label{eqn:minCutGeneral-caseI-one}
A=
\begin{blockarray}{ccccc}
                  & S\cap X          & S\setminus X   & t        & \comp{S+t}\\
\begin{block}{c(cccc)}
{S\cap X}         & *                &u_1\cdot v_1^T & u_1            & u_1\cdot v_2^T \\
&&&&\\
{S\setminus X}    & p_1\cdot q_1^T   & *              & u_2 &  u_2\cdot v_2^T\\       
&&&&\\
{t}         & q_1^T                & q_2^T & *              & *\\
&&&&\\
{\comp{S+t}}    & p_2\cdot q_1^T   & p_2\cdot q_2^T             & * & *  \\
\end{block}
\end{blockarray}
\end{equation}
where $v_1,p_1\in \mathbb{F}^{|S\setminus X|} $ and $v_2,p_2\in \mathbb{F}^{|\comp{S}|+1}$.
From $\cref{eqn:minCutGeneral-caseI-one},$  $$A[S\cap X,\comp{S\cap X} ] = u_1\cdot (v_1\mid 1\mid v_2)^T \text{ and } A[\comp{S\cap X}, S\cap X]= (p_1\mid 1\mid p_2)\cdot q_1^T.$$ This implies $S\cap X \subset S$ is a cut in $A$ which contradicts the minimality of $S$.

Now we prove the other part of the lemma. Suppose this $T$ is a singleton set, i.e. $T=\{t\}$ for some $t\in\comp{S}$. For the sake of contradiction, assume that there exists a cut $X$ in $A[S+t]$. Then, from the first part of the lemma, either $S\subseteq X$ or $S\subseteq \widetilde X$ where $\widetilde X=(S+t)\setminus X$. Without loss of generality, assume $S\subseteq X$. Then $|\widetilde X|\leq 1$. This is a contradiction since $X$ is a cut in $A[S+t]$. Therefore, $A[S\cup T]$ has no cut when $T$ is a singleton set.
\end{proof}

\begin{lemma} \label{lem:Blowerhalf}
Let $A$ be an $n\times n$ matrix over $\F$ with nonzero off-diagonal entries. Let $S\subseteq [n]$ be a cut in the matrix $A$ and $t\in S$, and suppose $X\subseteq \comp{S}$ is a cut in $A[\comp{S}+t]$. Then, $X$ is also a cut in the matrix $A$.
\end{lemma}
\begin{proof}
The off-diagonal entries of $A$ are nonzero. The sets $X$ and $S$ are cuts in  $A[S+t]$ and $A$, respectively. This implies that the matrix $A$ can be written as follows.
\[
A=
\begin{blockarray}{ccccc}
                    &  X             & \comp{S}\setminus X & t  & S-t \\
\begin{block}{c(cccc)}
X                   & *              & u_1\cdot v_1^T      & u_1     & u_1\cdot v_2^T \\
&&&&\\
\comp{S}\setminus X & p_1\cdot q_1^T & *                   & u_2     & u_2\cdot v_2^T \\
&&&&\\
t                   & q_1^T          & q_2^T               & *       & * \\
&&&&\\
S-t                 & p_2\cdot q_1^T & p_2\cdot q_2^T      & *       & * \\
\end{block}
\end{blockarray}
\]
where $$u_1, q_1\in\F^{|X|},\ \ \ \  v_1, u_2, p_1, q_2\in\F^{|\comp{S}\setminus X|},\ \ \ \ v_2, p_2\in\F^{|S|-1},$$
and `$*$'marked submatrices can be arbitrary. 
From the above structure of $A$, observe that $$A[X, \comp{X}]=u_1\cdot (v_1 \mid 1 \mid v_2)^T \text{ and } A[\comp{X},X]=(p_1 \mid 1\mid p_2)\cdot q_1^T.$$  Therefore, $X$ is a cut in $A$.
\end{proof}

Now, we get back to show the relationship between the minimal cuts of two PME matrices. First, we handle the case when the size of the minimal cut of $A$ is two.

\begin{lemma}\label{lem:cutsizetwo}
    Let $A$ and $B$ be two $n\times n$ matrices over field $\F$ with nonzero off-diagonal entries
    such that $A\PME B$ and $S=\{s_1,s_2\}$ is a cut in $A$ of size $2$. 
    Then either $S$ is a cut in $B$ or 
    for each $i \in \{1,2\}$, $B$ has a cut $X_i$, defined as
    \[X_i=\{t\in \comp{S}\mid A[S+t]\DE B[S+t]\}\cup \{s_i\},\]
    and $S$ is a cut in $\tw(B, X_i)$.
\end{lemma}

\begin{proof}
Without loss of generality, let $S = \{1,2\}$. 
Let $3 \leq t \leq n$. 
Since $A\PME B$, it follows that $A[\{1,2,t\}]\PME B[\{1,2,t\}]$. 
From~\cref{lem:noCut}, we have $A[\{1,2,t\}]\dge B[\{1,2,t\}]$.
Hence, there exists a diagonal matrix $D_t$ with $D_t[1,1]=1$ such that 
\[D_t A[\{1,2,t\}] D_t^{-1} = B[\{1,2,t\}] \text{ or} \] 
\[ D_t A[\{1,2,t\}] D_t^{-1} = B[\{1,2,t\}]^T . \] 
For any $t$ for which the former condition holds, we will have
\begin{eqnarray}
    \frac{B[1,t]}{B[2,t]}  &= \frac{A[1,t]A[2,1]}{A[2,t]B[2,1]}
    &= \frac{A[1,3]A[2,1]}{A[2,3]B[2,1]} \text{ and }     \label{eq:BcAc} \\
    \frac{B[t,1]}{B[t,2]}  &= \frac{A[t,1]A[1,2]}{A[t,2]B[1,2]} 
    &= \frac{A[3,1]A[1,2]}{A[3,2]B[1,2]} .
    \label{eq:BrAr}
\end{eqnarray}
The last two equalities hold because $\{1,2\}$ is a cut in $A$.
For any $t$ for which the later condition holds, we will have
\begin{eqnarray}
    \frac{B[1,t]}{B[2,t]}  &= \frac{A[t,1]A[1,2]}{A[t,2]B[2,1]} 
    &= \frac{A[3,1]A[1,2]}{A[3,2]B[2,1]}  \text{ and }     \label{eq:BcAr}  \\
    \frac{B[t,1]}{B[t,2]}  &= \frac{A[1,t]A[2,1]}{A[2,t]B[1,2]}
    &= \frac{A[1,3]A[2,1]}{A[2,3]B[1,2]}. 
    \label{eq:BrAc}
\end{eqnarray}
If equations (\ref{eq:BcAc}) and (\ref{eq:BrAr})
hold for every $3 \leq t \leq n$,  
or if equations (\ref{eq:BcAr}) and (\ref{eq:BrAc})
hold for every $3 \leq t \leq n$,  then 
$\{1,2\}$ will be a cut of $B$. 

Suppose that is not true. 
It follows that 
\[\frac{A[1,3]A[2,1]}{A[2,3]} \neq \frac{A[3,1]A[1,2]}{A[3,2]} .\]
Let $P \subseteq \{3,4, \dots, n\}$ be the set of indices 
for which equations (\ref{eq:BcAc}), (\ref{eq:BrAr})
hold and let $Q := \{3,4, \dots, n\} \setminus P$ be the set of indices 
for which equations (\ref{eq:BcAr}), (\ref{eq:BrAc})
hold.

We will show that $P \cup \{1\}$ is a cut in $B$. 
Consider two indices $s \in P$ and $t \in Q$. 
Consider the set $T = \{1,2, s, t\}$. 
Since equations (\ref{eq:BcAc}) and (\ref{eq:BrAr})
hold for $s$ and do not hold for $t$,
we have that 
\[\frac{B[1,t]}{B[2,t]}  \neq \frac{B[1,s]}{B[2,s]} \text{ or }
\frac{B[t,1]}{B[t,2]}  \neq \frac{B[s,1]}{B[s,2]} .
\]
Hence, $\{1,2\}$ is not a cut in $B[T]$ and $B[T] \not \DE A[T]$.
But, we have that $A[T] \PME B[T]$. 
Hence, there must be a cut in $B[T]$ (\cref{lem:size4Matrix}) 
In fact, $B[T]$ will have more than one cut. 
Because if $B[T]$ has a unique cut, say $\{1,t\}$,
then that will also be a unique cut of $A[T]$ (\cref{lem:size4Matrix}).
But, $A[T]$ has a cut $\{1,2\}$. 

So, we conclude that $B[T]$ has cuts $\{1,s\}$ and $\{1,t\}$. 
Hence, we can write 
\begin{eqnarray*}
        B[s,t]/B[1,t] &=&  B[s,2]/B[1,2] \text{ and } \\
        B[t,s]/B[2,s] &=&  B[t,1]/B[2,1] 
\end{eqnarray*}
Using these equations for every $s \in P$ and every $t \in Q$, 
we get that $X = P \cup \{1\}$ is a cut in $B$.  Similarly, we can show that $X'=P\cup \{2\}$ is a cut in $B$. From \cref{thm:cutmap}, $X\Delta X'=\{1,2\}$ is a cut of $\tw(B,X)$ and $\tw(B,X')$. 
% Note that $X$ and $X'$ can be computed in $\poly(n)$ $\F$-operations given $S$.
% Now, consider $B' := \tw(B,X)$.
% %
% From equations (\ref{eq:BcAc}), (\ref{eq:BrAr}), (\ref{eq:BcAr}),
% and (\ref{eq:BrAc}), we get that 
% for every $3 \leq t \leq n $,
% \begin{equation*}
%     \frac{B'[1,t]}{B'[2,t]}  = \frac{A[1,3]A[2,1]}{A[2,3]B[2,1]} \text{ and }   
%     \frac{B'[t,1]}{B'[t,2]}  = \frac{A[3,1]A[1,2]}{A[3,2]B[1,2]} .
% \end{equation*}
% Hence, $\{1,2\}$ is a cut in $B'$.
\end{proof}

In the following lemma, we show that a minimal cut of $A$ of size greater than two is also a cut of $B$.

\begin{lemma}
\label{lem:Greaterthan2commoncut}
Let $A$ and $B$ be two $n\times n$ matrices over $\F$ with nonzero off-diagonal entries. Let $A\PME B$, and $S\subseteq n$ be a minimal cut in $A$ of size greater than two. Then,  $S$ is also a cut in $B$. 
\end{lemma}
\begin{proof}
Let $s\in\comp{S}$. We show that for all $t\in \comp{S+s}$, the set $T_t:=\{s,t\}$ is a cut in $B[S+T_t]$. This will imply that $$B[S, t]=\alpha\cdot B[S, s] \text{ and } B[t, S]=\beta\cdot B[s, S]$$ for some non-zero $\alpha,\beta\in\F$. Hence, $S$ is a cut in $B$.

Since $S$ is a minimal cut in $A$ of size greater than two, from~\cref{lem:Blowerhalf}, there are no cuts in both the matrices $A[S+s]$ and $A[S+t]$. We have that $A\PME B$. Therefore, applying~\cref{lem:noCut}, $A[S+s]\dge B[S+s]$ and $A[S+t]\dge B[S+t]$. This implies that both $B[S+s]$ and $B[S+t]$ have no cuts. 

For the sake of contradiction, assume that $T_t$ is not a cut in $B[S+T_t]$. Note that $T_t$ is a cut in $A[S+T_t]$ of size two. Then, from~\cref{lem:cutsizetwo}, there exists a cut $X\subseteq S+T_t$ in the matrix $B[S+T_t]$ such that $s\in X$ but $t\notin X$. Since $|S+T_t|\geq 5$, either $|X|> 2$ or the size of $\widetilde X:=(S+T_t)\setminus X$ is greater than $2$. If $|X|>2$, then $X-s$ is a cut in $B[S+t]$ as $X$ is a cut of $B[S+T_t]$. Otherwise, $\widetilde X-t$ is a cut in $B[S+s]$. In both the cases, we have contradictions. Thus, $T_t$ is a cut in $B[S+T_t]$ for all $t\in \comp{S+s}$. This completes our proof.
\end{proof}

Let $S$ be a minimal cut of $A$. If $S$ is of size greater than two, then from \cref{lem:Greaterthan2commoncut}, $S$ is also a cut of $B$. Otherwise, if $|S| = 2$ and $S$ is not a cut of $B$, then from \cref{lem:cutsizetwo}, there exists a cut $X$ in $B$ such that $S$ is a cut of $\tw(B,X)$. Hence, from now on, we can assume that $A$ has a minimal cut $S$, which is also a cut of $B$.

Now, we go to the inductive step. Since $A\PME B$,  any principal submatrix of $A$ and the corresponding principal submatrix of $B$ are also principal minor equivalent. We fix one principal submatrix corresponding to set $\comp{S}+t$ for some $s$ in $S$ and try to get a cut sequence for $A$ and $B$ using the cut sequence for $A[\comp{S}+s]$ and $B[\comp{S}+s]$, which we get from induction hypothesis. For this, we show \cref{cl:smallToBig}. Before that, we state the following observation.

\begin{observation} \label{obs:cutTwist}
    Let $A$ be an $n\times n$ matrix with non-zero off-diagonal entries. Let $S\subset [n]$ and $X\subset S$ such that $X$ is a cut of $A[S]$. 
    Then,\begin{enumerate}
        \item If $X$ is a cut of $A$, then  $\tw(A[S],X) = \tw(A,X)[S]$.
        \item If $\comp{S}+X$ is a cut of $A$, then $\tw(A[S],X)=\tw(A,\comp{S}+X)[S]$
    \end{enumerate} 
\end{observation}

\begin{claim} \label{cl:smallToBig}
   Let $A$ and $B$ be two $n\times n$ matrices with non-zero off-diagonal entries and a common cut $S\subset [n]$  such that $A\PME B$. Let $s\in S$ and $A[\comp{S}+s]$ and  $B[\comp{S}+s]$ be cut-transpose equivalent with respect to cut sequence $( \Tilde{X}_1,\Tilde{X}_2,\dots, \Tilde{X}_{k})$.  Let $A_0=A$ and $X_i=\begin{cases}
       \Tilde{X}_i\cup S & \text{ if } s\in \Tilde{X}_i\\
       \Tilde{X}_i & \text{ otherwise}
   \end{cases}.$  Then,
   \begin{enumerate}
       \item     For each $i\in [k], X_i$ and $S$ are cuts of $A_{i-1}$ where $A_i=\tw(A_{i-1},X_i)$ and $A_k[\comp{S}+s]\DE B[\comp{S}+s]$.
        \item If $S$ is a minimal cut of $A$, then $S$ is also a minimal cut of $A_i$ for each $i\in [k]$. 
   \end{enumerate}
\end{claim}

\begin{proof}
  Given that there exists a sequence of matrices $(A[\comp{S}+s]=\Tilde{A}_0, \Tilde{A}_1,\dots,\Tilde{A}_{k})$ such that $\Tilde{A}_{i}=\tw(\Tilde {A}_{i-1},\Tilde{X}_i)$ where $\Tilde{X}_i$ is a cut of $\Tilde{A}_{i-1}$ and $\Tilde{A}_k\DE B[\comp{S}+s]$ or $B[\comp{S}+s]^T.$ 
    If $s\in \Tilde{X}_1$, then  $X_1=\Tilde{X}_1\cup S$. Since $S$ is a cut of $A=A_0$ and $(\comp{S}+s)\setminus \Tilde{X}_1$ is a cut of $A[\comp{S}+s]$, from \cref{lem:Blowerhalf}, $(\comp{S}+s)\setminus \Tilde{X}_1$ is a cut of $A_0$ which in turn implies its complement, that is, $X_1$ is a cut of $A_0$. 
    Since $X_1= \comp{\comp{S}+s}+ \Tilde{X}_1,$ from \cref{obs:cutTwist}, $\tw(A_0,X_1)[\comp{S}+s]=\tw(\Tilde{A}_0,\Tilde{X}_1)$.
    Note that in this case, $S\subseteq X_1$.
    
    If $s\notin \Tilde{X}_1$, then $X_1=\Tilde{X}_1$. Note that from \cref{lem:Blowerhalf}, $X_1$ is also a cut of $A_0$ and $X_1\subset \comp{S}$. 
    Hence from \cref{obs:cutTwist}, $\tw(A_0,X_1)[\comp{S}+s]=\tw(\Tilde{A}_0,\Tilde{X}_1)$. Given that $A_1=\tw(A_0,X_1)$ and $\Tilde{A}_1=\tw(\Tilde{A}_0,\Tilde{X}_1)$. 
    Hence, $A_1[\comp{S}+s]=\Tilde{A}_1$. Since, either $X_1\subset \comp{S}$ or $S\subset X_1$, from \cref{thm:cutmap}, $A_1$ also has cut $S$. The minimality of cut $S$ also follows from \cref{thm:cutmap} when $S$ is a minimal cut of $A$.  Iteratively, in a similar way, we can show that for each $i\in \{2,3,\dots,k\}, X_i$ and $S$ are cuts of $A_{i-1}$ and $A_i[\comp{S}+s]=\Tilde{A}_i$. Also, we can show that if $S$ is a minimal cut of $A$, then it is also a minimal cut of $A_i$ for each $i\in \{2,3,\dots,k\}$. 
\end{proof}

Since $A[\comp{S}+s]\PME B[\comp{S}+s]$, from induction hypothesis, we get that they are cut-transpose equivalent with respect to a cut sequence of length, say $k\leq 2(|\comp{S}|+1)\leq 2n-2$. Using \cref{cl:smallToBig}, we can get another matrix $A'$ from $A$ through a sequence of cut-transpose operations such that $A'[\comp{S}\cup \{s\}]\dge B[\comp{S}\cup \{s\}]$ and $S$ is a minimal cut of $A'$. We can go even further and assume that $A'[\comp{S}\cup \{s\}]\DE B[\comp{S}\cup \{s\}]$. This is because when $A'[\comp{S}\cup \{s\}]\DE B[\comp{S}\cup \{s\}]^T$, then we can work with $B^T$ instead of $B$.  The following lemma shows that we can get $B$ from $A'$ by using at most one cut-transpose operation. 

The number of cut transpose operations from $A$ to $A'$ is $k\leq 2n-2$. If $S$ has size $2$ and it is not a cut in $B$, then we need one cut-transpose operation from \cref{lem:cutsizetwo}. From \cref{lem:SplussToN}, we might need one more cut-transpose operation from $A'$ to $B$. This completes the proof of \cref{thm:main-one} by giving a cut-sequence of size at most $2n$ from $A$ to $B$.

\begin{lemma} \label{lem:SplussToN}
Let $A$ and $B$ be two $n\times n$ matrices over $\F$ with nonzero off-diagonal entries and $A\PME B$. Let $S\subseteq [n]$ be a minimal cut in $A$ and also a cut in $B$. Let $s\in S$ such that $A[\comp{S}+s]\DE B[\comp{S}+s]$. Then, either $A\DE B$ or $\tw(A,\comp{S})\DE B$.
\end{lemma}
\begin{proof}
Without loss of generality, assume that $S=[i]$ and $s=i$. Then, from the hypothesis, $B[\comp{S}+i]\PME A[\comp{S}+i]$. Since $S$ is a minimal cut in $A$, using~\cref{lem:minCutGeneral} and~\cref{lem:noCut}, there exists an $(i+1)\times (i+1)$ invertible diagonal matrix $D_1$ such that $D_1[i+1,i+1]=1$ and  $$ D_1\cdot A[[i+1]]\cdot D_1^{-1}=B[[i+1]] \text{ or  }B[[i+1]]^T.$$ From the hypothesis, there exists another $(n-i+1)\times (n-i+1)$ invertible diagonal matrix $D_2$ such that $D_2[i,i]=1$ and 
\begin{equation}
\label{eqn:SplussToN-one}
B[\comp{S}+i]=D_2\cdot A[\comp{S}+i]\cdot D_2^{-1}.
\end{equation}
We assume that the rows and columns of $D_2$ are indexed by $\comp{S}+i$. Next, we divide our proof into the following two cases.

\paragraph*{Case I:} In this case, we assume that 
\begin{equation}
\label{eqn:SplussToN-case-I-one}
D_1\cdot A[[i+1]]\cdot D_1^{-1}=B[[i+1]],
\end{equation}
and show $A\DE B$. Let $D$ be an $n\times n$ invertible matrix defined as follows: For all $k\in[n]$,
\[
D[k,k]=
\begin{cases}
D_1[k,k]                      &   \text{ if } k\in[i]\\
\frac{D_2[k,k]}{D_2[i+1,i+1]}  &   \text{ otherwise }.
\end{cases}
\]
We will show that $B$ is equal to $DAD^{-1}$. Since $S$ is a common cut in both the matrices $A$ and $B$, the rank-one submatrices $A[S, \comp{S}]$ and $B[S, \comp{S}]$ can be written as follows. 
\begin{align}
 A[S,\comp{S}] = A[S, i+1]\cdot \frac{A[i, \comp{S}]}{A[i, i+1]}  &\text{ and } A[\comp{S}, S] = A[\comp{S}, i]\cdot \frac{A[i+1, S]}{A[i+1, i]} \label{eqn:SplussToN-case-I-two-A}\\
 B[S,\comp{S}] = B[S, i+1]\cdot \frac{B[i, \comp{S}]}{B[i, i+1]}  &\text{ and } B[\comp{S}, S] = B[\comp{S}, i]\cdot \frac{B[i+1, S]}{B[i+1, i]} \label{eqn:SplussToN-case-I-two-B}
\end{align}
From~\cref{eqn:SplussToN-one} and~\cref{eqn:SplussToN-case-I-one},
\begin{eqnarray*}
B[i,i+1]        &=&  A[i,i+1]\cdot D_2^{-1}[i+1, i+1]\\
B[i, \comp{S}]  &=&  A[i, \comp{S}]\cdot D_2^{-1}[\comp{S}], \text{ and }\\ 
B[S, i+1]       &=&  D_1[S]\cdot A[S,i+1]
\end{eqnarray*}
Therefore, using the above equation and~\cref{eqn:SplussToN-case-I-two-B},
\begin{align*}
B[S, \comp{S}] &= D_1[S]\cdot A[S, i+1]\cdot \frac{D_2[i+1, i+1]\cdot A[i, \comp{S}]\cdot D_2^{-1}[\comp{S}]}{A[i, i+1]}\\
&=D[S]\cdot A[S,\comp{S}]\cdot D^{-1}[\comp{S}]
\end{align*}
Similarly, we can show that $$B[\comp{S}, S]=D[\comp{S}]\cdot A[\comp{S}, S]\cdot D^{-1}[S].$$ Applying~\cref{eqn:SplussToN-case-I-one} and~\cref{eqn:SplussToN-one}, we get that  
\begin{eqnarray*}
B[S]         &=&  D[S]\cdot A[S]\cdot D^{-1}[S]  \text{ and }\\
B[\comp{S}] &=&  D[\comp{S}]\cdot A[\comp{S}]\cdot D^{-1}[\comp{S}].
\end{eqnarray*} 
Thus, $B=DAD^{-1}$.

\paragraph*{Case II:} In this case, we assume that 
\begin{equation}
\label{eqn:SplussToN-case-II-one}
D_1\cdot A[[i+1]]\cdot D_1^{-1}=B[[i+1]]^T,
\end{equation}
 and show $B\DE \tw(A, \comp{S})$. Let $D$ be an $n\times n$ invertible diagonal matrix defined as follows: For all $k\in[n]$,
\[
D[k,k]=
\begin{cases}
D_1^{-1}[k,k]                      &   \text{ if } k\in[i]\\
\frac{D_2[k, k]}{D_2[i+1, i+1]}  &   \text{ otherwise }.
\end{cases}
\]
We will prove that $B$ is equal to $D\cdot \tw(A, \comp{S})\cdot D^{-1}$. Since $S$ is a cut, the matrix $A$ has the following structure.
\[
A=
\begin{blockarray}{ccc}
                    &  S             & \comp{S}\\
\begin{block}{c(cc)}
S                   & A[S]                                           & A[S, i+1]\cdot \frac{A[i,\comp{S}]}{A[i, i+1]}\\
&&\\
\comp{S}            & \frac{A[\comp{S}, i]}{A[i+1,i]}\cdot A[i+1, S] & A[\comp{S}] \\
\end{block}
\end{blockarray}.
\]
Thus, $\tw(A, \comp{S})$ can be written as follows.
\[
\tw(A,\comp{S})=
\begin{blockarray}{ccc}
                    &  S             & \comp{S}\\
\begin{block}{c(cc)}
S                   & A[S]^T                                           & A[i+1, S]^T\cdot \frac{A[i,\comp{S}]}{A[i, i+1]}\\
&&\\
\comp{S}            & \frac{A[\comp{S}, i]}{A[i+1,i]}\cdot A[S, i+1]^T & A[\comp{S}]  \\
\end{block}
\end{blockarray}.
\]
From~\cref{eqn:SplussToN-one} and~\cref{eqn:SplussToN-case-II-one}, we have that
\begin{eqnarray*}
B[i,i+1]        &=&   A[i,i+1]\cdot D_2^{-1}[i+1, i+1]\\
B[S, i+1]       &=&   D_1^{-1}[S]\cdot A[i+1, S]^T\\
B[i,\comp{S}]   &=&   A[i,\comp{S}]\cdot D_2^{-1}[\comp{S}].
\end{eqnarray*}
Using the above equation and~\cref{eqn:SplussToN-case-I-two-B}, 
\begin{align*}
B[S, \comp{S}] &= D_1^{-1}[S]\cdot A[ i+1, S]^T\cdot \frac{D_2[i+1, i+1]\cdot A[i, \comp{S}]\cdot D_2^{-1}[\comp{S}]}{A[i, i+1]}\\
&=D[S]\cdot \tw(A,\comp{S})[S,\comp{S}]\cdot D^{-1}[\comp{S}]
\end{align*}
Similarly, we can show that $$B[\comp{S}, S]=D[\comp{S}]\cdot \tw(A,\comp{S})[\comp{S}, S]\cdot D^{-1}[S].$$ Applying~\cref{eqn:SplussToN-case-II-one} and~\cref{eqn:SplussToN-one}, we get that
\begin{eqnarray*}
B[S]        &=&  D[S]\cdot \tw(A, \comp{S})[S]\cdot D^{-1}[S], \text{ and }\\
B[\comp{S}] &=&  D[\comp{S}]\cdot \tw(A, \comp{S})[\comp{S}]\cdot D^{-1}[\comp{S}].
\end{eqnarray*}
Thus, $B=D\cdot \tw(A,\comp{S})\cdot D^{-1}$.
\end{proof}

%%%%%%%%%%%%%%%%%%%%%%%%%%%%%%%%%%%%%%%%%%
%%%  Some useful Results  %%%%%%%%%%%%%%%%
%%%%%%%%%%%%%%%%%%%%%%%%%%%%%%%%%%%%%%%%%%

% \subsection{Some useful results on cut and cut-transpose operation}
% \label{subsec:proof-of-thm-one-some-useful-results}

%%%%%%%%%%%%  New Statement %%%%%%%%%%%%%%%%%%%%%%%%

%%%%%%%%%%%%  New Statement %%%%%%%%%%%%%%%%%%%%%%%%

\sgc{Proof reference appendix!}

%%%%%%%%%%%%  New Statement %%%%%%%%%%%%%%%%%%%%%%%%

%%%%%%%%%%%%  New Statement %%%%%%%%%%%%%%%%%%%%%%%%

%%%%%%%%%%%%  New Statement %%%%%%%%%%%%%%%%%%%%%%%%

%%%%%%%%%%%%  New Statement %%%%%%%%%%%%%%%%%%%%%%%%

%%%%%%%%%%%%  New Statement %%%%%%%%%%%%%%%%%%%%%%%%

%%%%%%%%%%%%  New Statement %%%%%%%%%%%%%%%%%%%%%%%%

\section{Algorithm for Principal Minor Equivalence Testing}

In this section, we give a proof of \cref{thm:main-two} by giving polynomial time algorithm for testing whether two matrices are principal minor equivalent. For reducible matrices, the problem reduces to smaller instances of principal minor equivalence testing for irreducible matrices from \cref{lem:redToIr}. Using \cref{ob:irreducible-matrix-and-directed-graph}, we can find these instances in polynomial time. Hence, it is sufficient to give a polynomial time algorithm for irreducible matrices. 

In \cref{algo:thm-main-one}, given two irreducible matrices $A$ and $B$ as input, we output a cut sequence with respect to which $A$ and $B$ are cut-transpose equivalent if $A\PME B$ otherwise, we output "No". The algorithm is directly based on the proof of characterization result. As mentioned earlier, we first reduce to an instance where all the off-diagonal entries are non-zero. The following claim describes how to get such an instance.

\begin{claim} \label{cl:YSubstitution}
Let $\F$ be a field of size greater than $10n^5$. Let $A$ and $B$ be two $n\times n$ irreducible matrices over $\F$. Then, in $\poly(n)$ time, we can find a diagonal matrix $D\in \F^{n\times n}$ such that $A+D$ and $B+D$ are nonsingular and all entries of $\adj{A+D}$ and $\adj
{B+D}$ are nonzero.
\end{claim}

\begin{remark}
When the size of the underlying field $\F$ is not greater than $10n^5$, we can construct an extension $\mathbb K$ of $\F$ such that $|\mathbb K|>10n^5$ and work with the larger field $\mathbb K$. We can also construct such an extension $\mathbb K$ in time $\poly(n)$.
\end{remark}

\begin{proof}
Given $A$ and $B$, we need to construct the following two types of diagonal matrices over $\F$ in $\poly(n)$ time. 
\begin{description}
    \item[Type I:] Find diagonal matrices $D_A$ and $D_B$ such that both $A+D_{A}$ and $B+D_B$ are nonsingular.
    
    \item[Type II:] For all $i,j\in[n]$, find diagonal matrices $A_{i,j}$ and $B_{i,j}$ such that $$\adj{A+A_{i,j}}[i,j]\neq 0 \text{ and } \adj{B+B_{i,j}}[i,j]\neq 0.$$
\end{description}

Before describing the construction of the above-mentioned diagonal matrices, we first discuss how to use them to get the diagonal matrix $D$ as promised in the claim. Using univariate polynomial interpolation, we combine all the above-mentioned diagonal matrices to a single $n\times n$ diagonal matrix $\widetilde D$ as follows. Let $T$ be a subset of $\F$ of size $2n^2+2$. Fix a bijection $\phi$ from $T$ to the set of diagonal matrices $$\{D_A,D_B\}\sqcup\{A_{i,j}, B_{i,j}\,\mid\, i, j\in[n]\}.$$ For each $i\in [n],$ let $P_i$ be a univariate polynomial in $y$ such that for each $e\in T,$ $P_i(e)=\phi(e)[i,i]$.  We can find $P_i$ in $\poly(n)$ time using Lagrange interpolation such that its degree is at most $2n^2+1$. Then, the diagonal matrix $\widetilde D$ is defined as $\widetilde D[i,i]=P_i$ for all $i\in[n]$. Observe that for each $e\in T,$  after substituting $y$ by $e$ in $\widetilde D$, we get $\phi(e)$. Thus, both $A+\widetilde D$ and $B+\widetilde D$ are nonsingular, and all entries of $\adj{A+\widetilde D}$ and $\adj{B+\widetilde D}$ are nonzero. In other words, the univariate polynomials $\det(A+\widetilde D)$, $\det(B+\widetilde D)$, $\adj{A+\widetilde D}[i,j]$ and $\adj{B+\widetilde D}[i,j]$ for each $i,j\in [n]$ are  nonzero. Note that each of these polynomials has a degree at most $(2n^2+1)\times n$. Now, we have found a matrix with univariate polynomials as its entries that satisfy the condition of our claim. 

Now consider the polynomial $$P=\det(A+\widetilde D)\cdot \det(B+\widetilde D)\cdot \prod_{i,j\in[n]}\adj{A+\widetilde D}[i,j]\cdot \adj{B+D}[i,j].$$ From~\cite{Csanky76, Berkowitz84}, we know that the determinant of an $n\times n$ matrix whose entries are univariate polynomials of at most $\poly(n)$ degree can be computed in $\poly(n)$ time. Thus, the polynomial $P$ can be computed in time $\poly(n)$. The degree of $P$ is at most $d=(2n^3+n)\times (2n^2+2)\leq 10n^5$. Therefore, for any subset $S\subseteq \F$ of size $d+1$, there exists an $a\in S$ such that $P(a)$ is nonzero. Find such a point $a$ in $S$. Given the polynomial $P$, this can be done in $\poly(n)$ time. This implies that all the polynomials in the product are also nonzero at $a\in S$. Hence, after substituting $y$ by $a$ in $\widetilde D$, we get a matrix $D$ that satisfies the condition of our claim. Next, we describe how to find diagonal matrices of \textbf{Type I} and \textbf{Type II} in $\poly(n)$ time. 

\paragraph*{Find Type I diagonal matrices.} Let $y$ be an indeterminate and $D'$ be a diagonal matrix with each diagonal entry is $y$. Then, the coefficient of $y^n$ in $\det(A+D')$ is one, hence, $\det(A+D')$ is nonzero. Compute the polynomial $\det(A+D')$. Since it is an univariate polynomial of degree $n$, in $\poly(n)$time, we can find a point $a\in\F$ such that the evaluation of $\det(A+D')$ at $a$ is nonzero. Then, the matrix $D_A$ we get by substituting $y=a$ in $D'$. Similarly, we can find $D_B$ in $\poly(n)$ time.

\paragraph*{Finding Type II diagonal matrices.} Let $G_A$ be the graph such that its vertex set in $[n]$ and $(i,j)$ is an edge inf $G_A$ if and only if $i\neq j$ and $A[i,j]\neq 0$.  Let $i,j\in[n]$. Since $A$ is irreducible, there exists a path from $i$ to $j$ in $G_A$. Let $$P=(i_0,i_1,i_2,\ldots, i_k) \text{ with } i_0=i,\ i_k=j,$$ be a shortest path from $i$ to $j$. In particular, when $i=j$, $P$ is $(i_0=i)$. We can compute such a path $P$ in time $\poly(n)$. Let $D'$ be a $n\times n$ diagonal matrix and $y$ be an indeterminate such that  for all $e\in[n]$,
\[
D'[e,e]=
\begin{cases}
0 & \text{if } e\in P\setminus \{i\}\\
y & \text{ otherwise.}
\end{cases}
\]
Next, following the proof of~\cite[Theorem~1]{Hartfiell84}, one can show that $$\adj{A+D'}[i,j]=\det\left((A+D')\left [[n]-j,[n]-i\right ]\right)\neq 0.$$
From~\cite{Csanky76, Berkowitz84}, we can compute $\adj{A+D'}[i,j]$ in time $\poly(n)$. It is a polynomial of degree at most $n-1$. Therefore, in $\poly(n)$ time, we can find a point $a\in\F$ such that the evaluation of $\adj{A+D'}[i,j]$ at $y=a$ is nonzero. Then, the diagonal matrix $A_{i,j}$ we get by substituting $y=a$ in $D'$. Similarly, we find $B_{i,j}$ for all $i,j\in[n]$.  
\end{proof}

Now, we show how to find a minimal cut in an irreducible matrix efficiently.

\begin{lemma} 
\label{lem:MinimalCutInPolynomialTime}
Let $A$ be an $n\times n$ irreducible matrix over a field $\F$. Then, we can test whether $A$ has a cut in $\poly(n)$ time. Moreover, if there exists a cut in $A$, then a minimal cut of $A$ can be computed using $\poly(n)$ time. 
\end{lemma}
\sgc{Take care bit complexity stuff, and change writing according to that!}
\begin{proof}
Let $2^{[n]}$ denote the set of all subsets of $[n]$. We first show that the functions $g_1,g_2:2^{[n]}\rightarrow\mathbb{Z}$, defined as $$\forall X\in2^{[n]},\ g_1(X):=\rank(A[X,\comp{X}]) \text{ and } g_2(X):=\rank(A[\comp{X}, X]),$$ are submodular functions. For each $i\in [n]$, let $V_i$ be the subspace of $\F^n$ spanned by the $i$th row vector of $A$ and the characteristic vector $\chi_i$ for the set $\{i\}$. Let $f:2^{[n]}\rightarrow\mathbb{Z}$ be the function defined as $$\forall X\in2^{[n]},\ f(X)=\dim\left(\sum_{e\in X}V_e\right).$$ It is not hard to verify that the function $f$ is a submodular function. Observe that a subset of row vectors of $A[X,\comp{X}]$ indexed by $T\subseteq X$ are linearly independent if and only if the set $\{\chi_e\, \mid\, e\in X \}\sqcup \{A[e',[n]]\mid e'\in T \}$ are linearly independent. Therefore, for all $X\in2^{[n]}$, $$f(X)=g_1(X)+|X|.$$ 
% Since $f$ and the cardinality function both are submodular function, 
Since $f$ is a submodular function, $g_1$ is a submodular function. Similarly, we can show that  $g_2$ is also a submodular function.

Since $g_1$ and $g_2$ are submodular functions, their sum $g=g_1+g_2$ is also a submodular function. For any set $T=\{t_1,t_2\}\sqcup \{t_3,t_4\}$ with four distinct elements from $[n]$, let $g_T$ be a function defined on subsets of $\comp{T}$ such that $$\forall X\subseteq \comp{T},\ g_T(X)=g(X\cup \{t_1,t_2\}).$$ 
For any $X\subseteq \comp{T}$ and $a,b\in\comp{T}$,
\begin{align*}
g_T(X\cup \{a\}) + g_T(X\cup \{b\}) &= g(X\cup \{a,t_1,t_2\}) + g(X\cup \{b,t_1,t_2\})\\
&\geq g(X\cup \{t_1,t_2\})+ g(X\cup \{a,b,t_1,t_2\})\ \  (\text{ submodularity of $g$})\\
&=g_T(X)+g_T(X\cup\{a,b\}).
\end{align*}
From the above, $g_T$ is a submodular function. Note that if there exists a cut $S$ in $A$ with $\{t_1,t_2\}\subseteq S$ and $\{t_3,t_4\}\subseteq \comp{S}$ if and only if the minimum value of function $g_T$ is at most 2. One can also observe that for any subset $X\subseteq \comp{T}$, $g_T(X)$ can be computed in $\poly(n)$ time. Thus, using the submodular minimization algorithm in~\cite[Chapter~45]{Sch03B}, we can compute the minimum the value of $g_T$ for any set $T=\{t_1,t_2\}\sqcup\{t_3,t_4\}$ of four distinct elements from $[n]$ in $\poly(n)$ time. There are at most $n^4$  such subsets $T$, and we can test whether $A$ has a cut by computing the minimum value of $g_T$ for all such possible subsets $T$. Thus, we can test whether $A$ has a cut in $\poly(n)$ time.

Now, we discuss how to find a minimal cut. For a subset $T=\{t_1,t_2\}\sqcup\{t_3,t_4\}$ with four distinct elements from $[n]$, let $g'_T$ be the function on subsets of $\comp{T}$ such that $$\forall X\subseteq \comp{T},\  g'_T(X)= (n+1)g_T(X)+|X|.$$ Since both $g_T$ and the cardinality function are submodular, $g'_T$ is also a submodular function. Next observe that for $X\subseteq \comp T$, the set $X$ minimizes $g'_T$ if and only if for any $S\subseteq [n]$ with $t_1,t_2\in S$ but $t_3,t_4\notin S$ the following holds:
\begin{enumerate}
    \item $g(X\cup\{t_1,t_2\})\leq g(S)$.
    \item if $g(X\cup\{t_1,t_2\})= g(S)$, then $|X\cup\{t_1,t_2\}|\leq |S|$.
\end{enumerate}
Therefore, a minimizing set of $g'_T$ gives a minimal cut that contains both $t_1\text{ and }t_2$ but not $t_3$ and $t_4$, if such a cut exists. Now, using~\cite[Theorem~45.1]{Sch03B}, we can compute minimizing sets for the submodular functions $g'_T$ for all possible subsets $T$, and thus, we get a minimal cut in $\poly(n)$ time if $A$ has a cut. 
\end{proof}

%%%%%%%%%%%%%%%%%%%%%%%%%%%%%%%%%%
%%%  Pseudocodes  %%%%%%%%%%%%%%%%
%%%%%%%%%%%%%%%%%%%%%%%%%%%%%%%%%%

\begin{algorithm}
\caption{Algorithm to test equal corresponding principal minors of two irreducible matrices}
\label{algo:thm-main-one}
\textbf{Input:} Two $n\times n$ irreducible matrices $A$ and $B$ over $\F$\\
\textbf{Output:} If $A\PME B$, then returns a cut sequence $\mathcal X$ of subsets of $[n]$ such that $A$, $B$ are cut-transpose equivalent with respect to $\mathcal X$. Otherwise, returns ``No''.\\
\begin{algorithmic}[1]
\State Using~\cref{cl:YSubstitution}, get $D$ and $A'\gets \adj{A+D}$ and $B'\gets \adj{B+D}$.
\State \Call{Finding-Cut-Sequence}{$A'$, $B'$, $[n]$}\\

\Function{Finding-Cut-Sequence}{$A$, $B$, $I$} 
\If{$|I|\leq 3$, \emph{or}, $A$ has \emph{no} cut}
    \If{$A$ is \emph{not} diagonally equivalent to $B$}
        \State \Return ``No''.
    \Else 
        \State \Return empty sequence.
    \EndIf
\Else
    \State $\widetilde B\leftarrow B$ \label{algo-thm-main-one-1}
    \State Using~\cref{lem:MinimalCutInPolynomialTime}, find a minimal cut $S\subseteq I$ in $A$.\label{algo-thm-main-one-base-case}
    \If{$|S|\geq 3$, \emph{and}, $S$ is \emph{not} a cut of $B$}
        \State \Return ``No''.
    \ElsIf{$|S|=2$, \emph{and}, $S$ is \emph{not} a cut of $B$}
        \State $X\gets$\Call{Min-Cut-size-Two}{$A$, $B$, $S$, $I$}
        \If{$X=$ ``No''}
            \State \Return ``No''.
        \EndIf
        \State $\widetilde B\leftarrow \tw(B,X)$\label{algo-thm-main-one-2}
    \EndIf  

    \State Let $s\in S$. 
    \State $\mathcal X'\leftarrow$\Call{Finding-Cut-Sequence}{$A(\comp{S}+s)$, $\widetilde B(\comp{S}+s)$, $\comp{S}+s$}.
    \If {$\mathcal X'=$ ``No''}
        \State \Return ``No''.
    \EndIf
    \State Let $\mathcal X'=(X_1', X_2',\ldots, X_k')$.
    \State $A_0\leftarrow A$.

    \For{$i=1$ to $k$}\label{algo-thm-main-one-3}
    \If{$s\in X_i'$}
        \State $X_i\gets X_i'\cup S$
    \Else
        \State $X_i\gets X_i'$
   \EndIf
    \State $A_i\leftarrow \tw(A_{i-1}, X_i)$.\label{algo-thm-main-one-5}
    \EndFor

% \State Let $t\in\comp{S}$.\sgc{Address Roshan's query}
% \If{$A_k(S+t) \DE \widetilde B(S+t)$}
    \If{$A_k\dge \widetilde B$}
        \State $\mathcal X \leftarrow (X_1, X_2,\ldots, X_k)$
    \ElsIf{$\tw(A_k, \comp{S}) \dge \widetilde B$}
        \State $\mathcal X \leftarrow (X_1, X_2,\ldots, X_k, \comp{S})$
    \Else \State \Return ``No''.\label{algo-thm-main-one-8}
    \EndIf

% \If{$A_k(S+t) \DE \widetilde B(S+t)$}\sgc{I suggest keep it this way; because ElsIF condition becomes similar}
%    \State $\mathcal X \leftarrow (X_1, X_2,\ldots, X_k)$
% \ElsIf{$A_k(S+t) \DE \widetilde B(S+t)^T$}
%    \State $\mathcal X \leftarrow (X_1, X_2,\ldots, X_k, \comp{S})$
% \Else 
%    \State Output ``No'', and stop execution.
% \EndIf

    \If{$|S|=2$, \emph{and}, $S$ is \emph{not} a cut of $B$}
        \State $\mathcal X\leftarrow (\mathcal X, X)$.\label{algo-thm-main-one-4}
    \EndIf

    \State \Return $\mathcal X$.
\EndIf
\EndFunction
\end{algorithmic}
\end{algorithm}

\begin{algorithm}[H]
\caption{Function for handling $|S|=2$ case in function \textsc{Cut-transpose} of~\cref{algo:thm-main-one}}
\label{algo:min-cut-two-algo}
\begin{algorithmic}[H]
\Function{Min-Cut-size-Two}{$A$, $B$, $I$, $S$}
\State $P\gets \emptyset$, and  $Q\gets \emptyset$
\State Let $s\in S$.
\For {$t\in I\setminus S$}
\If{$A(S+t)\DE B(S+t)$}
   \State $P\gets P\cup \{t\}$.
\ElsIf{$A(S+t)\DE B(S+t)^T$}
   \State $Q\gets Q\cup \{t\}$.
\Else
   \State \Return ``No''.
\EndIf
\EndFor

\State $X\gets P\cup \{s\}$. \label{algo2cutforsize2}

\If{$X$ is \emph{not} a cut of $B$}
   \State \Return ``No''.
\Else
  \State \Return $X$.
\EndIf
\EndFunction
\end{algorithmic}
\end{algorithm}

%%%%%%%%%%%%%%%%%%%%%%%%%%%%%%%%%%%%%%%%%%%
%%%  Proof of Correctness  %%%%%%%%%%%%%%%%
%%%%%%%%%%%%%%%%%%%%%%%%%%%%%%%%%%%%%%%%%%%

\subsection{Proof of Correctness of~\cref{algo:thm-main-one}}
\label{subsec:proof-of-correctness-algo1}

In the algorithm, we first find an invertible diagonal matrix $D$ such that $A+D$ and $B+D$ are invertible and off-diagonal entries of $\adj{A+D}$ and $\adj{B+D}$ are non-zero. The existence of such $D$ is guaranteed by \cref{cl:YSubstitution}. From \cref{lem:inversibleAdjugate}, $A\PME B$  if and only if $\adj{A+D}\PME \adj{B+D}$. Also, from \cref{lem:ResultForAplusDimpliesForA}, $\adj{A+D}$ and $\adj{B+D}$ are cut-transpose equivalent with respect to a cut sequence $\mathcal X$ if and only if $A$ and $B$ are cut-transpose equivalent with respect to $\mathcal X$. Hence, it is sufficient to show proof of the correctness of function \textsc{Finding-Cut-Sequence}  just for input matrices with non-zero off-diagonal entries. We do this in the following lemma.

\begin{lemma}
\label{lem:correctness-of-twist-function}
Let $A$ and $B$ be two matrices over $\F$ such that their rows and columns are indexed by elements in $I$. Let the off-diagonal entries of $A$ and $B$ be nonzero. Then, given $(A,B,I)$ as input to the function \textsc{Finding-Cut-Sequence} in~\cref{algo:thm-main-one}, it does the following:
\begin{enumerate}
\item If $A\PME B$, then it returns a sequence $\mathcal X$ of less than $2|I|$ many subsets of $I$ such that $A$ and $B$ are cut-transpose equivalent with respect to $\mathcal X$.
\item Otherwise, it returns ``No''.
\end{enumerate}
\end{lemma}
\begin{proof}
We use induction to prove the above lemma.
\paragraph*{Base case.} The base case of our induction is either $|I|\leq 3$, or $A$ has no cut. From~\cref{lem:noCut}, if $|I|\leq 3$ or $A$ has no cut, then $A\PME B$ if and only if either $A\dge B$. Therefore, for the base case,  the function \textsc{Finding-Cut-Sequence} in~\cref{algo:thm-main-one} returns ``No'' when $A\not\PME B$. Otherwise, it returns an empty sequence. 

\paragraph*{Inductive step.} In~\cref{algo-thm-main-one-base-case}, the function \textsc{Finding-Cut-Sequence} (in~\cref{algo:thm-main-one}) computes a minimal cut $S$ in the matrix $A$. If $|S|\geq 3$ and $A\PME B$, then from~\cref{lem:Greaterthan2commoncut}, $S$ is also a cut in $B$. This implies that if $S$ is not a cut in $B$, then $A\not\PME B$. Therefore, when $|S|\geq 3$ and $S$ is not a cut of $B$, the function \textsc{Finding-Cut-Sequence} returns ``No''. 

Now, consider the case when the size of the minimal cut $S$ is two, but it is not a cut in $B$. Then, ~\cref{algo:thm-main-one} calls the function \textsc{Min-cut-size-Two} of ~\cref{algo:min-cut-two-algo}. It returns ``No'' when either a minor corresponding to set $S+t$ where $t\in I\setminus S$ is not same for $A$ and $B$ or when $X$, defined in \cref{algo2cutforsize2}, is not a cut of $B$. If a minor is not the same, then obviously $A\not \PME B$. Otherwise from ~\cref{lem:cutsizetwo}, $A\not \PME B$ when $X$ is not a cut of $B$. If $A\PME B$, then from ~\cref{lem:cutsizetwo} $X$ is a cut of $B$ such that $S$ is a cut of $\tw(B,X)$.

Note that $\widetilde B$ is initially assigned to $B$ at~\cref{algo-thm-main-one-1} of~\cref{algo:thm-main-one}. If the function \textsc{Min-cut-size-Two} in~\cref{algo:min-cut-two-algo} returns a cut $X$, then $\widetilde B$ is reassigned to $\tw(B,X)$ at~\cref{algo-thm-main-one-2} of~\cref{algo:thm-main-one}. Thus, at the end of~\cref{algo-thm-main-one-2} of~\cref{algo:thm-main-one}, we have two matrices $A$ and $\widetilde B$ such that $S$ is a common cut of them, and also $S$ is a minimal cut in $A$.

Let $s\in S$, $M=A[\comp{S}+s]$, and $N=\widetilde B[\comp{S}+s]$. Since the cardinality of $S$ is at least two, the size of $\comp{S}+s$ is less than $|I|$. Therefore, from the induction hypothesis, the function \textsc{Finding-Cut-Sequence} on input $(M,N,\comp{S}+s)$ returns $\mathcal X'$ as follows:
\begin{enumerate}
\item If $M\PME N$, then $\mathcal X'=(X_1', X_2', \ldots, X_k')$ such that $k< 2|\comp{S}+s|$, $X_i'\subseteq \comp{S}+s$, and $\mathcal X'$ produces a sequence of matrices $(M=M_0, M_1, M_2, \ldots, M_k)$ satisfying the following: 
\begin{equation}
\label{eqn:correctness-of-twist-function-two}
\forall i\in[k],\ M_i=\tw(M_{i-1}, X_i') \text{ where } X_i' \text{ is a cut in } M_{i-1},\text{ and }M_k\dge N.
\end{equation}
\item Otherwise, $\mathcal X'=$``No''. 
\end{enumerate}
If $M\not\PME N$, then $A\not\PME \widetilde B$. Applying~\cref{lem:PMEwithTwist}, $A\not\PME \widetilde B$ implies that $A\not\PME B$. Therefore, when  $\mathcal X'=$``No'', the function \textsc{Finding-Cut-Sequence} also returns ``No''. 

Now assume that $M\PME N$, and $\mathcal X'=(X_1', X_2', \ldots, X_k')$ satisfies~\cref{eqn:correctness-of-twist-function-two}. Let $(X_1, X_2,\ldots, X_k)$ be the sequence of subsets of $I$ defined by the `for loop' in~\cref{algo-thm-main-one-3} of~\cref{algo:thm-main-one}. From this sequence of subsets, the function \textsc{Finding-Cut-Sequence} defines a sequence of matrices $(A=A_0, A_1, A_2,\ldots, A_k)$ such that $A_i=\tw(A_{i-1}, X_i)$ for all $i\in[k]$. From \cref{cl:smallToBig}, we get that $S$ is a minimal cut of $A_k$ and $A_k[\comp{S}+s]\dge \widetilde B[\comp{S}+s]$. Suppose $A_k[\comp{S}+s]\DE \widetilde B[\comp{S}+s]$, then from~\cref{lem:SplussToN}, we have either $A_k\DE \widetilde B$ or $\tw(A_k, \comp{S})\DE \widetilde B$ when $A_k\PME \widetilde B$. Suppose $A_k[\comp{S}+s]\DE \widetilde B[\comp{S}+s]^T$ or in other words $A_k[\comp{S}+s]\DE \widetilde B^T[\comp{S}+s]$. Then, from~\cref{lem:SplussToN}, we have either $A_k\DE \widetilde B^T$ or $\tw(A_k, \comp{S})\DE \widetilde B^T$ when $A_k\PME \widetilde B^T\PME B$.

 Hence, if $A_k\not \dge \widetilde B$ and $\tw(A_k, \comp{S})\not \dge \widetilde B$, it follows that $A\not\PME \widetilde B$, in which case, the function \textsc{Finding-Cut-Sequence} returns ``No''. When $A_k\dge \widetilde B$, the function \textsc{Finding-Cut-Sequence} defines $\mathcal X$ as $(X_1, X_2,\ldots, X_k)$, and when $\tw(A_k, \comp{S})\dge \widetilde B$, it defines $\mathcal X$ as $(X_1, X_2,\ldots, X_k, \comp{S})$. Therefore, at the end of~\cref{algo-thm-main-one-3} in~\cref{algo:thm-main-one}, we obtain a  sequence $\mathcal X$ such that $A$ and $B$ are cut-transpose equivalent with respect to $\mathcal X$.

Note that if $S$ is a minimal cut in $A$ of size $2$ and it is not a cut in $B$, then $\widetilde B$ is defined as $\tw(B, X)$. This implies that $B=\tw(\widetilde B, X)$. Therefore, in~\cref{algo-thm-main-one-4} of~\cref{algo:thm-main-one}, $\mathcal X$ is updated by appending $X$ at its end. Thus, we finally have a sequence $\mathcal X$ of subsets of $I$ such that $A$ and $B$ are cut-transpose equivalent with respect to $\mathcal{X}$.

\end{proof}

\subsection{Time complexity of~\cref{algo:thm-main-one}}
\label{subsec:tim-complexity-algo-thm-main-one}
\sgc{Take care of the bit complexity stuff, and change writing according to that!}
~\cref{algo:thm-main-one} first computes a diagonal matrix $D$ such that both $A+D$ and $B+D$ are nonsingular and all the entries of $A'=\adj{A+D}$ and $B'=\adj{B+D}$ are nonzero.~\cref{cl:YSubstitution} ensures that we can compute such a diagonal matrix $D$ in $\poly(n)$ time. Then, it calls the function \textsc{Finding-Cut-Sequence}. The function makes at most one recursive call to itself such that the size of the input matrices reduces by at least one. The matrix operations like finding the cut-transpose of a matrix with respect to a given cut, finding a minimal cut (\cref{lem:MinimalCutInPolynomialTime}), testing diagonal equivalence (\cref{cl:DE-in-poly-time}), and others can be performed in polynomial time in terms of the matrix size. Hence, the overall runtime of the algorithm is polynomial.
\section{PIT for Sum of two DET1}
\label{sec:PIT}

In this section, we show \cref{thm:main-three}. Given two sequences of $n\times n$ matrices $(A_0,A_1,\ldots, A_m)$ and $(B_0,B_1,\ldots, B_m)$ over a field $\F$ such that the rank of $A_i$ and $B_i$ is at most $1$ for $1\leq i\leq n$, the goal is to decide whether two polynomials $P_1=\det(A_0+A_1y_1+\ldots+A_my_m)\text{ and } P_2=\det(B_0+B_1y_1+\ldots+B_my_m)$ are the same in $\poly(m,n)$ time. First, we consider the case when $A_0$ and $B_0$ are the zero matrix. Then, we reduce the general case where there are no constraints on $A_0$ and $B_0$ to this case. Then, we give a polynomial time reduction from this problem to the problem of equivalence testing of principal minors of two $m\times m$ matrices. For integers $p\text{ and }q$, let $0_p$ and $0_{p,q}$ denote the $p\times p$ and $p\times q$ matrix, respectively, with all zeros.

 % Note that if $m<n$, then $P_1$ and $P_2$ are both zero. Hence, without loss of generality, we assume that $m\geq n$.  
   \subsection[A0=B0=0]{$\boldsymbol{A_0=B_0=0_n}.$} 
 Let $A_j=u_{1,j}\cdot v_{1,j}^T$ and $B_j=u_{2,j}\cdot v_{2,j}^T$ for each $j\in [m]$ where $u_{1,j},v_{1,j},u_{2,j},v_{2,j}\in \F^n$.  Let $U_i, V_i$ be $n\times m$ matrices such that their $j$th column are $u_{i,j}$ and $v_{i,j}$, respectively, for $i\in \{1,2\}$ and $j\in [m]$. Let $Y$ be an $m\times m$ diagonal matrix with indeterminate $y_i$ as the $i$th diagonal entry. Then, \begin{equation}\label{eq:AequalUXV} A_1y_1+\ldots+A_my_m= U_1YV_1^T \text{ and } B_1y_1+\ldots+B_my_m=U_2YV_2^T.\end{equation}
 For a subset $T$ of $[m],$ let  $y_T=\prod_{e\in T}y_e$, $U_{i,T}=U_i[[n],T]$ and $V_{i,T}=V_i[[n],T]$ for $i\in \{1,2\}$.
 Using the Cauchy-Binet formula for multiplying two rectangular matrices, 
 \begin{equation*}\det(U_iYV_i^T)= \sum\limits_{T\subseteq [m], |T|=n}\left(\det(U_{i,T})\det(V_{i,T}) y_T \right)  \:\: \text{ for } i\in \{1,2\}.\end{equation*} Hence, by comparing coefficients of monomials of $P_1$ and $P_2$, we get
 \begin{equation}\label{eq:Det1ToUVE}P_1=P_2 \iff \det(U_{1,T})\det(V_{1,T}) = \det(U_{2,T})\det(V_{2,T}) \:\: \forall T\subseteq [m] \text{ with }|T|=n.\end{equation}
Now, we discuss how to test the latter part mentioned above. First, we find a set $T$ of size $[n]$ such that $\det(U_{1,T})\det(V_{1,T})$ is non-zero using a matroid intersection algorithm for matroids represented by $U_1$ and $V_1$ in $\poly(m,n)$ time. If such $T$ doesn't exist, then $P_1=0$. Similarly, we can check whether $P_2$ is zero and decide whether $P_1=P_2$. Suppose such a set $T$ exists and without loss of generality, let $T=[n]$. If $\det(U_{1,[n]})\det(V_{1,[n]})\neq \det(U_{2,[n]})\det(V_{2,[n]}),$ then $P_1\neq P_2$ from \cref{eq:Det1ToUVE}.

Suppose $\det(U_{1,[n]})\det(V_{1,[n]})= \det(U_{2,[n]})\det(V_{2,[n]})$. Now, we have to check this for other sets $T$ of size $n$. Let 
$U_i'=U_{i,[n]}^{-1}\cdot U_i \text{ and } V_i'=V_{i,[n]}^{-1}\cdot V_i\text{ for }i=1,2.$ Since $U_i=U_{i,[n]}\cdot U_i', V_i=V_{i,[n]}\cdot V_i'$ for $i=1,2$ and $\det(U_{1,[n]})\det(V_{1,[n]})= \det(U_{2,[n]})\det(V_{2,[n]})$, for any set $T$ of size $n$,
\begin{equation} \label{eq:UvToU'V'}
    \det(U_{1,T})\det(V_{1,T}) = \det(U_{2,T})\det(V_{2,T}) \iff  \det(U_{1,T}')\det(V_{1,T}') = \det(U_{2,T}')\det(V_{2,T}')
\end{equation}

Note that $U_{i,[n]}'=V_{i,[n]}'=I_n$. For $i=1,2$, let $\widehat U_i$ and $\widehat V_i$ be the $n\times (m-n)$ matrices defined as $ U_i'[[n],[m]\setminus [n]]$ and $ V_i' [[n],[m]\setminus [n]]$, respectively. For $i\in \{1,2\}$ and a set $T=T_1'\sqcup T_2'$  of size $n$ with $T_1'\subseteq [n], T_2'\subseteq [m]-[n]$ such that $T_2'=\{n+e\mid e\in T_2\}$ where $T_2\subseteq [m-n]$, 
\begin{equation} \label{eq:U'V'minors}\det(U_{i,T}')= \sigma(T) \det(U_i'[[n]\setminus T_1', T_2']) \text{ and } \det(V_{i,T}')= \sigma(T) \det(V_i'[[n]\setminus T_1', T_2']) \end{equation}  
where $\sigma:\binom{[m]}{n}\xrightarrow{}\{1,-1\}$ is some sign function on $n$ sized subsets of $[m]$. Since $U_i'[[n]\setminus T_1', T_2']= \widehat U_i[T_1,T_2]$ and $V_i'[[n]\setminus T_1', T_2']= \widehat V_i[T_1,T_2] $  where $T_1=[n]\setminus T_1'$, using \cref{eq:Det1ToUVE,eq:UvToU'V',eq:U'V'minors} we get
\begin{eqnarray} \label{eq:P1P2}
    P_1=P_2 & \iff \det(\widehat U_1[T_1,T_2])\det(\widehat V_1[T_1,T_2]) = \det(\widehat U_2[T_1,T_2])\det(\widehat V_2[T_1,T_2])\\
     & \text{ for each }  T_1\subseteq[n], T_2\subseteq [m-n] \text{ with } |T_1|=|T_2| \nonumber 
\end{eqnarray}

Let $A$ and $B$ be the $m\times m$ matrices defined as follows: 
\[
A=
\left[
\begin{array}{c c c |  c}
 &  0_{m-n} &  & \widehat V_1^T\\
 &              &  & \\
\hline
 &-\widehat U_1 &  &0_{n}
\end{array}
\right]\quad\text{and}\quad
B=
\left[
\begin{array}{c c c |  c}
 &  0_{m-n} &  & \widehat V_2^T\\
 &              &  & \\
\hline
 &-\widehat U_2 &  &0_{n}
\end{array}
\right].
\]
Let us consider the principal minors of $A$ and $B$. If a set $T$ is a subset of $[m-n]$ or $[m]-[m-n]$, then the corresponding principal minors of both $A$ and $B$ are zero. Consider a set $T=T_1'\sqcup T_2$ such that $T_2\subseteq [m-n]$ and $T_1'\subseteq [m]-[m-n]$ such that $T_1'=\{m-n+e \mid e\in T_1\}$ where $T_1\subseteq [n]$. Then, \[A[T]=\left[
\begin{array}{c c c |  c}
 &  0_{|T_2|} &  & \widehat V_1[T_1,T_2]^T\\
 &              &  & \\
\hline
 &-\widehat U_1[T_1,T_2] &  &0_{|T_1|}
\end{array}
\right] \quad\text{and}\quad
B[T]=\left[
\begin{array}{c c c |  c}
 &  0_{|T_2|} &  & \widehat V_2[T_1,T_2]^T\\
 &              &  & \\
\hline
 &-\widehat U_2[T_1,T_2] &  &0_{|T_1|}
\end{array}
\right].\] 
Note that if $|T_1|\neq |T_2|$, then both $\det(A[T])$ and $\det(B[T])$ are zero. If $|T_1|=|T_2|$, then 
\begin{equation}\label{eq:minorDefAB}
    \det(A[T])=  \det(\widehat U_1[T_1,T_2])\det(\widehat V_1[T_1,T_2]) ;  \det(B[T])=  \det(\widehat U_2[T_1,T_2])\det(\widehat V_2[T_1,T_2]).
\end{equation}
From above discussion and \cref{eq:minorDefAB}, 
\begin{eqnarray} \label{eq:PMEAB}
    A\PME B\iff \det(\widehat U_1[T_1,T_2])\det(\widehat V_1[T_1,T_2]) =  \det(\widehat U_2[T_1,T_2])\det(\widehat V_2[T_1,T_2])\\  \forall T_1\subseteq[n], T_2\subseteq [m-n] \text{ with } |T_1|=|T_2|. \nonumber
\end{eqnarray}

From Eq. \ref{eq:P1P2} and Eq. \ref{eq:PMEAB}, $P_1=P_2\iff A\PME B$. Note that $A$ and $B$ can be computed in  $\poly(m,n)$ time. From \cref{thm:main-one}, we can check whether $A\PME B$ in $\poly(m)$ time.
This completes the proof of \cref{thm:main-three} when $A_0$ and $B_0$ are the zero matrix.

\subsection[A0nq0]{No constraint on $\boldsymbol{A_0}$ and $\boldsymbol{B_0}$}

From \cref{eq:AequalUXV}, $P_1=\det(A_0+U_1YV_1^T)$ and $P_2=\det(B_0+U_2YV_2)^T$. From \cite[Lemma 4.3]{DBLP:conf/stoc/GurjarT17} $P_1=\det(C_1)$ and $P_2=\det(C_2)$ such that 
\[C_1=  \begin{pmatrix}
    I_m & Y & 0_{m,n}\\
    0_m & I_m & V_1^T\\
    U_1 & 0_{n,m} &  A_0
\end{pmatrix} \text{ and } C_2 =  \begin{pmatrix}
    I_m & Y & 0_{m,n}\\
    0_m & I_m & V_2^T\\
    U_2 & 0_{n,m} &  A_0
\end{pmatrix}.\]
  If we compute $P_i=\det(C_i)$ using the Generalized Laplace Theorem by fixing the first $m$ rows, we get that $P_i$ is multilinear and the coefficient of $y_T$ for a subset $T$ of $[m]$ is $$\sigma(T)\det(C_i[[2m+n]\setminus [m],\phi(T)])$$ where $\sigma$ is some sign function depending on set $T$ and $\phi:2^{[m]}\xrightarrow{}\binom{[2m+n]}{m+n}$ such that $\phi(T)= T\cup \{e+m\mid e\notin T \}\cup ([2m+n]\setminus [2m]).$  Hence,
  \begin{equation} \label{eq:C1C2}
  P_1=P_2\iff  \det(C_1[\comp{[m]},\phi(T)]) = \det(C_2[\comp{[m]},\phi(T)]) \:\: \forall T\subseteq[m].  \end{equation}
Let $V$ be the $(m+n)\times (2m+n)$ matrix
 $\begin{pmatrix}
     I_m & I_m & 0_{m,n}\\
     0_{n,m} & 0_{n,m} & I_n
 \end{pmatrix}.$ 
 For any set $T'\subset [2m+n]$ of size $m+n$ which is not in image of $\phi,$ $\det(V[[m+n],T'])=0$ and if $T'$ belongs to image of map $\phi,$ then $\det(V[[m+n],T'])=1$. Hence,
 \begin{eqnarray} \label{eq:VC1C2}
      \det(C_1[\comp{[m]},\phi(T)]) = \det(C_2[\comp{[m]},\phi(T)]) \:\: \forall T\subseteq[m] \iff  \nonumber\\ 
      \det(C_1[\comp{[m]},T'])\det(V_{T'})= \det(C_2[\comp{[m]},T'])\det(V_{T'})\:\: \forall T'\in \binom{[2m+n]}{m+n}.
 \end{eqnarray}
 Note that the right-hand side of the above equation is similar to \cref{eq:Det1ToUVE}. Hence, using similar arguments from the previous section, checking the later part of Eq. \ref{eq:VC1C2} can be reduced to checking whether principal minors of two $(2m+n)\times (2m+n)$ matrices (that can be computed in $\poly(n)$ time) are the same. Hence, from \cref{eq:C1C2}, checking whether $P_1=P_2$ reduces to checking whether principal minors of two $(2m+n)\times (2m+n)$ matrices are the same. This completes the proof of \cref{thm:main-three}.

% \paragraph{Reduction to Principal minor equality problem.}~First find a common base $B$ of the linear matriods $U_1$ and $V_1$.Then check whether $\det(U_{1,B})\cdot \det(V_{1,B})=\det(U_{2,B})\cdot \det(V_{2,B})$. If not, then $\det(U_1XV_1^T)\neq\det(U_2XV_2^T)$. Wlog, assume that $B=\{1,2,\ldots,r\}$. Otherwise, consider the matrices $$U_i'=U_{i,B}^{-1}\cdot U_i \text{ and } V_i'=V_{i,B}^{-1}\cdot V_i\text{ for }i=1,2.$$ For $i=1,2$, let $\widehat U_i$ and $\widehat V_i$ be the $r\times (n-r)$ matrices defined as follows: $\widehat U_i$ and $\widehat V_i$ are the submatrices formed by the columns of $U_i'$ and $V_i'$ indexed by $[n]\setminus [r]$, respectively. 
% \begin{theorem}
% \label{thm:principal-minor-equi-LMI}
% In Problem~\ref{prob:equi-LMI}, $\det(U_1XV_1^T)=\det(U_2XV_2^T)$ if and only if $A$ and $B$ has equal corresponding principal minors.
% \end{theorem}

% \input{Sum-of-LMI}
% \input{Pseudocodes}

\bibliographystyle{alpha} 
\bibliography{biblio.bib}
\appendix
\section{Missing Proofs from~\cref{sec:prelim}}
\label{appendix:missing-proofs-from-prprelim}
\repstatement{lem:PMEwithTwist}{Let $A$ be an $n\times n$ irreducible matrix over a field $\F$ with a cut $X\subset [n]$. Then, $A\PME \tw(A,X)$.}

\begin{proof}
Let $q$ and $u$ be the first non-zero row of $A[X,\comp{X}]$ and the first non-zero column of $A[\comp{X},X]$, respectively. Without loss of generality, we can assume that the cut $X$ is a prefix of the index set and hence $A$ can be written as follows: 
\[A=
\begin{blockarray}{ccc}
         & X          & \comp{X}\\
\begin{block}{c(cc)}
X        & M          & p\cdot q^T\\
&&\\
\comp{X} & u\cdot v^T & N\\
\end{block}
\end{blockarray},
\]
where $p, v\in\F^{|\comp{X}|}$ and $q, u\in\F^{|X|}$. Then, from \cref{def:twist-operation},\sgc{May need reference fixing!} 
\[\tw(A,X)=
\begin{blockarray}{ccc}
         & X          & \comp{X}\\
\begin{block}{c(cc)}
X        & M          & p\cdot u^T\\
&&\\
\comp{X} & q\cdot v^T & N^T\\
\end{block}
\end{blockarray}.
\]
% Without loss of generality, \Red{ we can assume that $X=[i]$ for some  $i< n$.} 
% we can assume that $X$ is a prefix of $[n]$.
Let $S\subseteq [n]$. Observe that if $S$ is a subset of $X$ or $\comp{X}$, then $\det(A(S))= \det(\tw(A,X)(S))$. Now consider that  $S=S_1\sqcup S_2$ such that $S_1$ and $S_2$ are nonempty subsets of  $X$ and $\comp{X}$, respectively.  Next, we prove that $\det(A(S))=\det(\tw(A,X)(S))$.
% let $v_T$ denote a vector of dimension $|T|$ with $v_T[e]=v[e]$ for each $e\in T$. 

Assume that the coordinates $p, v$ are indexed by $X$ and the coordinates of $q, u$ are indexed by $\comp{X}$. By $p_{S_1}, q_{S_2}, u_{S_1}$ and $v_{S_2}$, we denote the projection of the respective vectors on the respective coordinates.  Let $A'=A[S]$ and $B'=\tw(A, X)[S]$. Then, 
\[
A'=
\begin{pmatrix}
M[S_1]                 & p_{S_1}\cdot q_{S_2}^T\\
&\\
u_{S_2}\cdot v_{S_1}^T & N[S_2]
\end{pmatrix}
,\:\:\text{ and }\:\: 
B'=
\begin{pmatrix}
M[S_1]                 & p_{S_1}\cdot u_{S_2}^T\\
&\\
q_{S_2}\cdot v_{S_1}^T & N[S_2]^T
\end{pmatrix}.
\]
If either of $p_{S_1}, q_{S_2}, u_{S_2}$, or $v_{S_1}$ is the zero vector, then $$\det(A')=\det(B')=\det(M[S_1])\det(N[S_2]).$$ Next, assume that all of them are nonzero.

Let $\ell=|S|$, $k=|S_1|$, and $K=[k]$. Suppose that the rows and columns of $A'$ and $B'$ are indexed by $[\ell]$, and the rows and columns of $M[S_1]$ are indexed by $K$.  For each $i\in K$, let $M_i$ denote the $k\times k$ matrix obtained by removing $i$th column of $M[S_1]$ and appending $p_{S_1}$ as the $k$th column. For $j\in \comp{K}$, let $N_j$ denote the $(l-k)\times (l-k)$ matrix obtained by removing $j$th column of $N[S_2]$ and adding $u_{S_2}$ as the first column. Using the Generalized Laplace Theorem (see~\cite[Theorem~3.1]{Ahmadieh23}),  $\det(A')$ can be written as follows.
$$\det(A') = \sum_{T\subseteq [\ell],\, |T|=k} (-1)^{\sum K+\sum T}\det(A'[K, T])\det(A'[\comp{K},\comp{T}]).$$
Note that for all $T\subset [\ell]$ with $|T\cap \comp{K}|\geq 2$, the submatrix $A'[K, T])$ is not full rank since $\rank(A'[K,\comp{K}])\leq 1$. Therefore, for all such $T\subseteq [\ell]$ with $|T|=k$,  $\det(A'[K, T])=0$. This implies that 
$$\det(A') = \det(M[S_1])\det(N[S_2]) + \sum_{i\in K,j\in \comp{K}}(-1)^{j-i}\det(A'[K, K-i+j])\det(A'[\comp{K},\comp{K}+i-j]).$$
Observe that for all $i\in K$ and $j\in \comp{K}$, $$\det(A'[K, K-i+j])=q_{S_2}[j-k]\det(M_i)\text{, and }\det(A[\comp{K}, \comp{K}+i-j])=v_{S_1}[i]\det(N_j).$$ Therefore, 

\begin{align*}
\det(A') &= \det(M[S_1])\det(N[S_2]) + \sum_{i\in K,j\in \comp{K}}(-1)^{j-i}v_{S_1}[i]q_{S_2}[j-k]\det(M_i)\det(N_j)\\
           &= \det(M[S_1])\det(N[S_2])+\left(\sum_{i\in K}(-1)^{i} v_{S_1}[i]\det(M_i)\right)\left(\sum_{j\in \comp{K}}(-1)^{j} q_{S_2}[j-k]\det(N_j)\right).
\end{align*}
Similarly, using the Generalized Laplace Theorem (see~\cite[Theorem~3.1]{Ahmadieh23}), we compute $\det(B')$. For each $j\in \comp{K}$, let $\widetilde{N}_j$ denote the matrix obtained by removing $j$th row of $N[S_2]$ and adding $q_{S_2}$ as the first row. Then, 
$$\det(B')=\det(M[S_1])\det(N[S_2])+\left(\sum_{i\in K}(-1)^{i} v_{S_1}[i]\det(M_i)\right)\left(\sum_{j\in \comp{K}}(-1)^{j} u_{S_2}[j-k]\det(\widetilde N_j)\right).$$ 
Let $P$ be the following $(|S_2|+1)\times (|S_2|+1)$ matrix
\[
P=
\begin{pmatrix}
0       & q_{S_2}^T\\
u_{S_2} & N[S_2]
\end{pmatrix}.
\]
Then, \[(-1)^k\det(P)=  \sum_{j\in\comp{K}}(-1)^j q_{S_2}[j-k]\det(N_j)= \sum_{j\in \comp{K}}(-1)^ju_{S_2}[j-k]\det(\widetilde{N}_j).\] The above equalities follow from the expression for computing the determinant of $P$ by fixing the first row and the first column of $P$, respectively. Hence, $\det(A[S])=\det(\tw(A,X)[S])$.
\end{proof}

{\em 
\paragraph*{Notations.} Suppose that $A$ is an $n\times n$ irreducible matrix over a field $\F$ with a cut $X\subset [n]$. Let $q$ and $u$ be the first non-zero row of $A[X,\comp{X}]$ and the first non-zero column of $A[\comp{X},X]$, respectively. Then, the matrix $A$ has the following structure:
\[A=
\begin{blockarray}{ccc}
         & X          & \comp{X}\\
\begin{block}{c(cc)}
X        & M         & p\cdot q^T\\
&&\\
\comp{X} & u\cdot v^T & N\\
\end{block}
\end{blockarray}\,,
\]
where $p, v\in\F^{|X|}$ and $q, u\in\F^{|\comp{X}|}$. Without loss of generality, assume that $X=[\ell]$. Then, $\comp{X}=[n]\setminus [\ell]$. For each $i\in X$, let 
\begin{enumerate}
\item $M_i^C$ denote the $\ell\times \ell$ matrix obtained by removing $i$th column of $M$ and appending $p$ as the $\ell$th column.
\item $M_i^R$ denote the $\ell\times \ell$ matrix obtained by removing $i$th row of $M$ and appending $v^T$ as the $\ell$th row.
\end{enumerate}
For each $j\in \comp{X}$, let $\comp{X}_j$ denote the set $\comp{X}-j$.
Let $p^A\in\F^{|X|}$ and $q^A\in\F^{|\comp{X}|}$ be defined as follows: for all $i\in X$ and $j\in\comp{X}$,
\begin{equation}
\label{eqn:vector1-submatrix-in-adj}
p_A[i]=(-1)^{\ell+i+1}\det(M_i^C)\ \  \text{ and }\ \  q_A[j-\ell] = \sum_{k\in\comp{X}}(-1)^{k+j}q[k-\ell]\cdot \det(N[\comp{X}_j, \comp{X}_k]). 
\end{equation}
Similarly, let $v^A\in\F^{|X|}$ and $u^A\in\F^{|\comp{X}|}$ be defined as follows: for all $i\in X$ and $j\in \comp{X}$,
\begin{equation}
\label{eqn:vector2-submatrix-in-adj}
v_A[i] = (-1)^{\ell+i+1}\det(M_i^R)\ \  \text{ and }\ \  u_A[j-\ell]=\sum_{k\in\comp{X}}(-1)^{k+j}u[k-\ell]\cdot \det(N[\comp{X}_k, \comp{X}_j]).
\end{equation}
Based on the above notations, we have the following claim.
}
\begin{claim}
\label{cl:submatrix-in-adj}
Considering the notations defined above, 
\[
A^{\mathrm{adj}}[X, \comp{X}]=p_A\cdot q_A^T\ \  \text{ and }\ \  A^{\mathrm{adj}}[\comp{X}, X]=u_A\cdot v_A^T,
\]
where $p_A, q_A, u_A$ and $v_A$ are defined as~\cref{eqn:vector1-submatrix-in-adj} and~\cref{eqn:vector2-submatrix-in-adj}.
\end{claim}
\begin{proof}
The proof of the above claim will closely follow the proof of~\cite[Lemma~4.5]{Ahmadieh23}. For all $i\in[n]$, let $(n)_i$ denote the set $[n]\setminus \{i\}$. Then, for all $i\in X$  and $j\in\comp{X}$, $$A^{\mathrm{adj}}[i,j]=(-1)^{i+j}\cdot \det\left(A[(n)_j, (n)_i]\right).$$
For any $T\subseteq (n)_i$, let $m_T$ denote the number of elements in $T$ that are greater than $i$. For example, if $T$ contains exactly one element from $\comp{X}$ and $|T|=\ell$, then $m_T=\ell-i+1$. Now, using the Generalized Laplace Theorem (see~\cite[Theorem~3.1]{Ahmadieh23}), for all $i\in X$  and $j\in\comp{X}$,
\[
\det\left(A[(n)_j, (n)_i]\right) =\sum_{T\subseteq (n)_i,\ |T|=\ell}(-1)^{m_T+\sum X+\sum T}\det(A[X, T])\cdot \det(A[\comp{X}_j, \comp{T}_i]). 
\]
Since the rank of $A[X, \comp{X}]$ is at most one, for all $\ell$-size subsets $T$ of $[n]_i$ with $|T\cap \comp{X}|\geq 2$, the value of $\det(A[X, T])$ is zero. Thus, from the above equation, 
\begin{align*}
\det\left(A[(n)_j, (n)_i]\right) &= (-1)^{\ell-i+1}\cdot \sum_{k\in\comp{X}}(-1)^{k-i}\det(M_i^C)\cdot q[k-\ell]\cdot \det(A[\comp{X}_j, \comp{X}_k])\\
&=(-1)^{\ell+1}\det(M_i^C) \cdot\sum_{k\in\comp{X}}(-1)^{k} q[k-\ell]\cdot \det(A[\comp{X}_j, \comp{X}_k])
\end{align*}
This implies that, for all $i\in X$ and $j\in\comp{X}$, 
\begin{align*}
A^{\mathrm{adj}}[i,j] &= (-1)^{\ell+i+1}\det(M_i^C)\cdot\sum_{k\in\comp{X}}(-1)^{k+j} q[k-\ell]\cdot \det(A[\comp{X}_j, \comp{X}_k])\\
&= p_A[i]\cdot q_A[j-\ell].
\end{align*}
Similarly, we can show that for all $i\in X$ and $j\in\comp{X}$, $$A^{\mathrm{adj}}[j,i] = u_A[j-\ell]\cdot v_A[i].$$ This completes the proof of the above claim.
\end{proof}

\repstatement{lem:AdjugateCut}{\em Let $A$ be an $n\times n$ matrix over a field $\F$. Let $D$ be an $n\times n$ diagonal matrix over $\F$ such that $A+D$ is non-singular. Then, $A$ and $\adj{A+D}$ have the same set of cuts.}
\begin{proof}
Observe that $A$ and $A+D$ have the same set of cuts. Next, we show that $A+D$ and $\adj{A+D}$ have the same set of cuts, implying that $A$ and $\adj{A+D}$ have the same set of cuts. From~\cref{cl:submatrix-in-adj}, any cut in $A+D$ is also a cut in $\adj{A+D}$. For the converse direction, assume that $X$ is a cut in $\adj{A+D}$. Therefore, $X$ is also a cut in $(A+D)^{-1}$ since $A+D$ is non-singular, $$\adj{\adj{A+D}}=\det(A+D)^{n-2}(A+D).$$ Thus, again using~\cref{cl:submatrix-in-adj}, $X$ is also a cut in $A+D$.
\end{proof}

\repstatement{lem:AdjugateTwistEquality}{
Let $A$ be an $n\times n$ irreducible matrix over a field $\F$. Then, a cut $X\subseteq [n]$ of $A$ is also a cut of $A^{\mathrm{adj}}$ and $$\tw(A, X)^{\mathrm{adj}}\DE \tw(A^{\mathrm{adj}}, X).$$ 
}
\begin{proof} 
Without loss of generality, assume that $X=[\ell]$. Then, $\comp{X}=[n]\setminus [\ell]$. Let $q$ and $u$ be the first non-zero row of $A[X,\comp{X}]$ and the first non-zero column of $A[\comp{X},X]$, respectively. Then, $A$ can be written as follows. 
\[A=
\begin{blockarray}{ccc}
         & X          & \comp{X}\\
\begin{block}{c(cc)}
X        & M         & p\cdot q^T\\
&&\\
\comp{X} & u\cdot v^T & N\\
\end{block}
\end{blockarray}\,,
\]
where $p, v\in\F^{|X|}$ and $q, u\in\F^{|\comp{X}|}$. From~\cref{def:twist-operation}, \sgc{May need some fixing!}
\[\widetilde A=\tw(A, X)=
\begin{blockarray}{ccc}
         & X          & \comp{X}\\
\begin{block}{c(cc)}
X        & M         & p\cdot u^T\\
&&\\
\comp{X} & q\cdot v^T & N^T\\
\end{block}
\end{blockarray}\,.
\]
Repeating some notations from the above, for each $i\in X$, let 
\begin{enumerate}
\item $M_i^C$ denote the $\ell\times \ell$ matrix obtained by removing $i$th column of $M$ and appending $p$ as the $\ell$th column.
\item $M_i^R$ denote the $\ell\times \ell$ matrix obtained by removing $i$th row of $M$ and appending $v^T$ as the $\ell$th row.
\end{enumerate}
For each $j\in \comp{X}$, let $\comp{X}_j$ denote the set $\comp{X}-j$. Then, using~\cref{cl:submatrix-in-adj}, 
\[A^{\mathrm{adj}}=
\begin{blockarray}{ccc}
         & X          & \comp{X}\\
\begin{block}{c(cc)}
X        & A^{\mathrm{adj}}[X]     & p_A\cdot q_A^T\\
&&\\
\comp{X} & u_A\cdot v_A^T           & A^{\mathrm{adj}}[\comp{X}]\\
\end{block}
\end{blockarray}\,,
\]
where $p_A\in\F^{|X|}$ and $q_A\in\F^{|\comp{X}|}$ are defined as~\cref{eqn:vector1-submatrix-in-adj} and $v_A\in\F^{|X|}$ and $u_A\in\F^{|\comp{X}|}$ are defined as~\cref{eqn:vector2-submatrix-in-adj}. Let $A'$ be the following matrix. 
\[A'=
\begin{blockarray}{ccc}
         & X          & \comp{X}\\
\begin{block}{c(cc)}
X        & A^{\mathrm{adj}}[X]     & p_A\cdot u_A^T\\
&&\\
\comp{X} & q_A\cdot v_A^T           & A^{\mathrm{adj}}[\comp{X}]^T\\
\end{block}
\end{blockarray}\,.
\]
From \cref{remark:twist}, $\tw(A^{\mathrm{adj}}, X)\DE A'$ On the other hand, again applying~\cref{cl:submatrix-in-adj}, 
\[\widetilde A^{\mathrm{adj}}=
\begin{blockarray}{ccc}
         & X          & \comp{X}\\
\begin{block}{c(cc)}
X        & \widetilde A^{\mathrm{adj}}[X]                       & \tilde p\cdot \tilde u^T\\
&&\\
\comp{X} & \tilde q\cdot \tilde v^T              & \widetilde A^{\mathrm{adj}}[\comp{X}]\\
\end{block}
\end{blockarray}\,,
\]
where like~\cref{eqn:vector1-submatrix-in-adj} and~\cref{eqn:vector2-submatrix-in-adj}, $\tilde p, \tilde v\in\F^{|X|}$ and $\tilde q, \tilde u\in\F^{|\comp{X}|}$ are defined as follows: For $i\in X$ and $j\in\comp{X}$,
\begin{align*}
\tilde p[i]     &=(-1)^{\ell+i+1}\det(M_i^C)\ \  \text { and }\ \ \tilde u[j-\ell]=\sum_{k\in\comp{X}}(-1)^{k+j}u[k-\ell]\cdot \det(N^T[\comp{X}_j, \comp{X}_k])\\
\tilde v[i] &=(-1)^{\ell+i+1}\det(M_i^R)\ \ \text{ and }\ \ \tilde q[j-\ell]=\sum_{k\in\comp{X}}(-1)^{k+j}q[k-\ell]\cdot \det(N^T[\comp{X}_k, \comp{X}_j]). 
\end{align*}
Now we show that $A'=\widetilde A^{\mathrm{adj}}$, which in turn implies $\tw(A^{\mathrm{adj}}, X)\DE \widetilde A^{\mathrm{adj}}$. For all $i\in[n]$, let $(n)_i$ denote the set $[n]\setminus \{i\}$. Next, we divide our proof into three cases.
\begin{enumerate}
\item Assume that $i, j\in X$. Then, $$A'[i, j] = A^{\mathrm{adj}}[i, j]=\det(A[(n)_j, (n)_i]).$$ On the other hand, $$\widetilde A^{\mathrm{adj}}[i,j] = \det(\widetilde A[(n)_j, (n)_i]).$$ If $|X|=2$, then $\widetilde A[(n)_j, (n)_i]= A[(n)_j, (n)_i]^T$. Otherwise, $\comp{X}$ is also a cut of $ A[(n)_j, (n)_i]$ and $$ \widetilde A[(n)_j, (n)_i]\DE \tw(A[(n)_j, (n)_i],\comp{X})^T.$$ Therefore, applying~\cref{lem:PMEwithTwist}, $\det(\widetilde A[(n)_j, (n)_i])= \det(A[(n)_j, (n)_i])$ and hence, $A'[i,j]=\widetilde A^{\mathrm{adj}}[i,j]$ for all $i,j\in X$.

\item Assume that $i, j\in \comp{X}$. Then, $$A'[i, j]=A^{\mathrm{adj}}[j,i]=\det(A[(n)_i, (n)_j]).$$ On the other hand, $$\widetilde A^{\mathrm{adj}}[i,j] = \det(\widetilde A[(n)_j, (n)_i]).$$ If $|\comp{X}|=2$,  $\widetilde A[(n)_j, (n)_i]= A[(n)_i, (n)_j]$.  Otherwise, $X$ is also a cut of $ A[(n)_i, (n)_j]$ and $$ \widetilde A[(n)_j, (n)_i] \DE \tw(A[(n)_i, (n)_j],X).$$  Therefore, again applying~\cref{lem:PMEwithTwist}, $A'[i,j]=\widetilde A^{\mathrm{adj}}[i,j]$ for all $i,j\in \comp{X}$.

\item Assume that $i\in X$ and $j\in \comp{X}$. Then,
\begin{align*}
A'[i,j] &= p_A[i]\cdot u_A[j-\ell]\\
 &= (-1)^{\ell+i+1}\det(M_i^C) \cdot \sum_{k\in\comp{X}}(-1)^{k+j}u[k-\ell]\cdot \det(N[\comp{X}_k, \comp{X}_j])\\
 &= (-1)^{\ell+i+1}\det(M_i^C) \cdot \sum_{k\in\comp{X}}(-1)^{k+j}u[k-\ell]\cdot \det(N^T[\comp{X}_j, \comp{X}_k])\\
 &= \tilde p[i]\cdot \tilde u[j] = \widetilde A^{\mathrm{adj}}[i,j].
\end{align*}
Similarly, we can show that $A'[j,i]=\widetilde A^{\mathrm{adj}}[j,i]$.
\end{enumerate}
This completes the proof of our lemma.
\end{proof}

% \begin{claim} \label{cl:size3Matrix}
%     Let $A$ be a $3\times 3$ matrix with all off-diagonal entries as non-zero. Let $B$ be another $3\times 3$ matrix such that $A\PME B$. Then, either $A\DE B$ or $A\DE B^T$.
% \end{claim}
% \begin{proof}
% \sgc{Proof is already available in \cite[Theorem~3]{Hartfiell84}. I think saying it claim is not correct. Better to use Lemma. Also, add size two. Because, that used in the induction of~\cref{lem:correctness-of-twist-function}}
% \end{proof}

\section{Proof of \cref{lem:size4Matrix}}

\label{appendix:otherProofs}

\em{First, we present the following lemma, which provides an alternative characterization for principal minor equivalence testing and will be used in our proof.  For an $n\times n$ matrix $A$, let $G_A$ be a directed graph on $n$ vertices such that there exists a directed edge $(i,j)$ if and only if $A[i,j]\neq 0$. For a directed cycle $C$ of $G_A$, let $w_A(C)$ denote the weight of the cycle defined as $\prod_{(i,j)\in C}A[i,j]$. }

\begin{lemma} \label{lem:PMEEqualWeightedCycle}
    Let $A$ and $B$ be two $n\times n$ matrices. Then $A\PME B$ if and only if for each subset $S\subseteq [n]$ , the sum of weights of directed Hamiltonian cycles is the same for subgraphs $G_A[S]$ and $G_B[S].$ \end{lemma} 
\begin{proof}
    We show this by induction on the size of subsets. The base case, when the size of the subset is one, is trivial. Suppose the statement is true for subsets of size $k$. Let $S$ be a subset of size $k+1$ and $\det(A[\emptyset])=\det(B[\emptyset])=1$ and $\mathcal{C}_A$ and $\mathcal{C}_B$ denote the set of Hamiltonian cycles of $G_A[S]$ and $G_B[S]$, respectively. Then, 
    \begin{equation} \label{eq:AS}\det(A[S])=\sum\limits_{T\neq \emptyset, T\subseteq S } \pm (\det(A[S\setminus T])\prod_{i\in T}A[i,i]) \pm\sum_{C\in \mathcal{C}_A} w_A(C).\end{equation} 
   \begin{equation}\label{eq:BS} \det(B[S])=\sum\limits_{T\neq \emptyset, T\subseteq S } \pm (\det(B[S\setminus T])\prod_{i\in T}B[i,i]) \pm\sum_{C'\in \mathcal{C}_B} w_B(C').\end{equation} The backward direction follows directly from \cref{eq:AS,eq:BS} as all the principal minors of $A$ and $B$ are the same, and the signs of corresponding summands are the same in \cref{eq:AS,eq:BS}.
   Now, we show the forward direction. For any non-empty subset $T$, if we consider the sub-matrices $A[S\setminus T]$ and $B[S\setminus T]$, for each subset $T'\subseteq S\setminus T$ the sum of the weights of the directed Hamiltonian cycles in $G_A[T']$ and $G_B[T']$ are same.
   Hence, $\det(A[S\setminus T])=\det(B[S\setminus T])$ by induction hypothesis. Also, the signs are the same in \cref{eq:AS} and \cref{eq:BS} as it just depends on the size of the subsets $T$. This, along with the fact that the sum of weights of the Hamiltonian cycle in $G_A[S]$ and $G_B[S]$ are the same, implies $\det(A[S])=\det(B[S]).$

%    follows from the induction hypothesis on strict subsets of $S$ that implies $\det(A[S\setminus T])=\det(B[S\setminus T])$ for non-empty subsets $T$ of $S$; and the fact that the sum of the weights of Hamiltonian cycles is same in both $G_A[S]$ and $G_B[S]$. 

%    Now, we show the backward direction. For any non-empty subset $T$, if we consider the matrices $A[S\setminus T]$ and $B[S\setminus T]$, for each subset $T'\subseteq S\setminus T$ the sum of the weights of the cycles in $G_A[T']$ and $G_B[T']$ are same. Hence, the sum of weights of cycles in $G_A[S]$ and $G_B[S]$ are the same.
 \end{proof}

 % Following is an immediate corollary of the above lemma.
 % \begin{corollary} \label{cor:AddDiagonal}
 %     Let $A$ and $B$ be two $n\times n$ matrices such that $A\PME B$. Let $D$ be an $n\times n$ diagonal matrix. Then, $A+D\PME B+D$. 
 % \end{corollary}
%  \begin{proof}
%      Since $A\PME B$, $(A+D)[i]=(B+D)[i]$ for each $i\in [n]$. This implies the sum of weights of the Hamiltonian cycle in $G_A[i]$ and $G_B[i]$ are the same. Since $A\PME B$, from \cref{lem:PMEEqualWeightedCycle},  for each subset $S\subseteq [n]$ , the sum of weights of directed Hamiltonian cycles is the same for subgraphs $G_A[S]$ and $G_B[S].$ Note that $G_A$ and $G_{A+D}$ have the same set of cycles, and any cycle of length greater than one has the same weight in $G_A$ and $G_{A+D}$. Similarly, $G_B$ and $G_{B+D}$ have the same set of cycles, and any cycle of length greater than one has the same weight in $G_B$ and $G_{B+D}$. This implies that for each subset $S\subseteq [n]$ , the sum of weights of directed Hamiltonian cycles is the same for subgraphs $G_{A+D}[S]$ and $G_{B+D}[S].$ Hence, from \cref{lem:PMEEqualWeightedCycle}, $A+D\PME B+D$.
% \end{proof}
\repstatement{lem:size4Matrix} 
Let $A$ be a $4\times 4$ matrix over $\F$ with all off-diagonal entries are nonzero. Let $B$ be another $4\times 4$ matrix over $\F$ such that $A\PME B$. Then, one of the following two holds:
\begin{enumerate}
\item $A\dge B$.
\item There exists a common cut in $A$ and $B$. Furthermore, for any common cut $X$ of $A$ and $B$, $\tw(A,X)\dge B$.
\end{enumerate}
\sgc{Instead of saying it a `Claim', call it a `Lemma'. The exact reference of the first item in the above statement is shown in \cite[Theorem~4]{Hartfiell84}.}
\begin{proof}
    When $A$ does not have any cut, then from \cref{lem:noCut} $A\DE B$ or $A\DE B^T$. Suppose $A$ has a cut. Then, $B$ must have some cut; otherwise, since $A\PME B$, \cref{lem:noCut} would imply that $A$ has no cut, which is a contradiction.  First, we show that $B$ must have a cut that is common to $A$ using contradiction. Suppose this is not true. Without loss of generality, assume that $A$ has cut $\{1,2\}$ and $B$ has cut $\{1,3\}$. Since $A \PME B$, $A'=D_1AD_1^{-1}\PME D_2BD_2^{-1}=B'$ for any non-singular diagonal matrices $D_1$ and $D_2$. Since off-diagonal entries are non-zero, we can choose $D_1$ such that $A'[1,3]=A'[1,4]=A'[2,3]=1$. Since $A$ and $A'$ has same cuts, $\rank(A'[\{1,2\}, \{3,4\}])=1$ and hence $A'[2,4]=1$. If the above claim is true for $A'$ and $B$ then it is also true for $A$ and $B$. Hence, without loss of generality, we can assume that each entry of $A[\{1,2\},\{3,4\}]$ is one. Similarly, we can assume that each entry of $B[\{1,3\},\{2,4\}]$ is one. Let $A$ be the following matrix with non-zero off-diagonal entries.
     \begin{equation*}A=\begin{pmatrix}
                * & a & 1 & 1\\
                b & * & 1 & 1\\
                c &  dc & * & e \\
                f &  df  & g & *
                \end{pmatrix}\end{equation*}
 Since $A\PME B$, we can represent $B$ as follows by making its size two principal minors the same as of $A$.
     \begin{equation*}B=\begin{pmatrix}
                * & 1 & h & 1\\
                  ab & * & dc & i\\
                \frac{c}{h} & 1 & * & 1\\
                f &  \frac{df}{i} & eg & *
                \end{pmatrix}\end{equation*}
Since $\rank(B[\{2,4\},\{1,3\}])=1$, $abeg=dcf$. After substituting $g$ with $\frac{dcf}{abe}$, 
\begin{equation*}
    A=\begin{pmatrix}
                * & a & 1 & 1\\
                b & * & 1 & 1\\
                c &  dc & * & e \\
                f &  df  & \frac{dcf}{abe} & *
                \end{pmatrix} \text{ and } B= \begin{pmatrix}
                * & 1 & h & 1\\
                  ab & * & dc & i\\
                \frac{c}{h} & 1 & * & 1\\
                f &  \frac{df}{i} & \frac{dcf}{ab} & *
                \end{pmatrix}.
\end{equation*}
Since $A\PME B$, from \cref{lem:PMEEqualWeightedCycle}, we get the following equations by equating the sum of weights of Hamiltonian cycles in $G_A[S]$ and $G_B[S]$ such that $|S|=3$.
\begin{align}
   \label{eq:41}   ac + bdc = \frac{dc^2}{h} + abh \implies \left(a-\frac{dc}{h}\right)(c-bh) =0  & & [S=\{1,2,3\}] \\
   \label{eq:42} af+bdf = if +\frac{abdf}{i} \implies f(a-i)\left(1-\frac{bd}{i}\right)=0 & &[S=\{1,2,4\}]\\ 
   \label{eq:43} ef+\frac{dc^2f}{abe} =  hf + \frac{dc^2f}{abh} \implies f(e-h)\left(1-\frac{dc^2}{abeh}\right)=0 & & [S=\{1,3,4\}]\\
   \label{eq:44} edf +\frac{d^2c^2f}{abe} =\frac{d^2cf}{i}+\frac{dcif}{ab} \implies fd\left(e-\frac{cd}{i}\right)\left(1-\frac{ci}{abe}\right) =0 & & [S=\{2,3,4\}].
\end{align}
Since the off-diagonal entries are non-zero, in each equation, at least one of the last two factors must be zero. This gives 16 different possibilities because there are four equations. Now, we show that each of these possibilities would imply a contradiction.

Note that $\left(a-\frac{dc}{h}\right)=0$ and $(e-h)=0$ together implies $ae=dc$ which in turn implies $\{1,3\}$ is a cut of $A$ which is a contradiction. Hence, $\left(a-\frac{dc}{h}\right)$ and $e-h$ can't be zero together. Similarly, $\left(a-\frac{dc}{h}\right)$ and $a-i$ can't be zero together as it implies $cd=hi$. This implies $\{1,2\}$ is a cut of $B$, which contradicts the condition of no common cut. Hence, if $\left(a-\frac{dc}{h}\right)=0$ then  $\left(1-\frac{bd}{i}\right)$ and $\left(1-\frac{dc^2}{abeh}\right)$ must be zero to make \cref{eq:42,eq:43} zero. $\left(a-\frac{dc}{h}\right)=0$  and $\left(1-\frac{bd}{i}\right)=0$ imply $abh=ci$ which in turn implies $\{1,4\}$ is a cut of $B$. Similarly, $\left(a-\frac{dc}{h}\right)=0$ and $\left(1-\frac{dc^2}{abeh}\right)=0$ imply $c=be$. This implies $\{1,4\}$ is also a cut of $A$, which contradicts the no common cut condition. Note that $\left(a-\frac{dc}{h}\right)=0$
 always led to a contradiction. Hence, it must be that $(c-bh) =0$  so that \cref{eq:41} is satisfied.

 Now, we show that $(c-bh) =0$ would also always lead to a contradiction. If $\left(1-\frac{bd}{i}\right)$ is also zero, then $cd=hi$ which implies $\{1,2\}$ is a cut of $B$ which is a contradiction. Similarly, if $\left(1-\frac{dc^2}{abeh}\right)=0$ then $cd=ae$ which implies $\{1,3\}$ is a cut of $A$ which contradicts the no common cut condition. Hence, to satisfy \cref{eq:42,eq:43}, it must be that $(e-h)$ and $(a-i)$ is equal to zero along with $(c-bh)$. However, then $\{1,4\}$ becomes a cut of both $A$ and $B$. Hence, $(c-bh)$ also can't be zero. This contradicts $\cref{eq:41}$. Hence, if $A$ has a cut and $A\PME B$, then $A$ and $B$ must have a common cut.

Without loss of generality, let that common cut be $\{1,2\}$. Using earlier arguments, we can represent $A$ and $B$ as follows by making all entries of $A[\{1,2\},\{3,4\}]$ and $B[\{1,2\},\{3,4\}]$ one and equating size two principal minors.
\[    A=\begin{pmatrix}
                * & a & 1 & 1\\
                b & * & 1 & 1\\
                c &  dc & * & e \\
                f &  df  & g & *
                \end{pmatrix} \text{ and } B= \begin{pmatrix}
                * & h & 1 & 1\\
                  \frac{ab}{h} & * & 1 & 1\\
                c & dc & * & i\\
                f &  df & \frac{ge}{i} & *
                \end{pmatrix}.
\]
Since $A\PME B$, from \cref{lem:PMEEqualWeightedCycle}, we get the following equations by equating the sum of weights of Hamiltonian cycles in $G_A[S]$ and $G_B[S]$ for $S=\{1,2,3\}$ and $S=\{1,3,4\}$.
\begin{align}
     \label{eq:51}   ac+bdc= hc+\frac{abdc}{h} \implies c(a-h)\left(1-\frac{bd}{h}\right)=0  & & [S=\{1,2,3\}]\\
   % \label{eq:52} af+bdf = hf +\frac{abdf}{h} \implies f(a-h)\left(1-\frac{bd}{h}\right)=0 & &[S=\{1,2,4\}]\\ 
    \label{eq:53} ef+cg =  if + \frac{cge}{i} \implies (e-i)\left(f-\frac{cg}{i}\right)=0 & & [S=\{1,3,4\}]
   % \label{eq:54} edf +dcg = idf+\frac{dcge}{i} \implies d(e-i)\left(f-\frac{cg}{i}\right) =0 & & [S=\{2,3,4\}].
\end{align}
\cref{eq:51,eq:53} together implies the following four possible cases. $(a-h)=0$ and $(e-i)=0$ implies $A=B$. Hence, $A\DE B$. 

If $\left(1-\frac{bd}{h}\right)=0$ and $\left(f-\frac{cg}{i}\right)=0$, then
\[B= \begin{pmatrix}
                * & bd & 1 & 1\\
                  \frac{a}{d} & * & 1 & 1\\
                c & dc & * & \frac{cg}{f}\\
                f &  df & \frac{fe}{c} & *
                \end{pmatrix} = \begin{pmatrix}
                1 & 0 & 0 & 0\\
                0 & \frac{1}{d} & 0 & 0\\
                0 &  0 & c & 0 \\
                0 &  0  & 0 & f
                \end{pmatrix} \begin{pmatrix}
                * & b & c & f\\
                a & * & dc & df\\
                1 &  1 & * & g \\
                1 &  1  & e & *
                \end{pmatrix} \begin{pmatrix}
                1 & 0 & 0 & 0\\
                0 & d & 0 & 0\\
                0 &  0 & \frac{1}{c} & 0 \\
                0 &  0  & 0 & \frac{1}{f} 
                \end{pmatrix} = DA^TD^{-1}. 
\] Hence, in this case, we get $A\DE B^T$.

If $(a-h)=0$ and $\left(f-\frac{cg}{i}\right)=0$, then 
\[B= \begin{pmatrix}
                * & a & 1 & 1\\
                  b & * & 1 & 1\\
                c & dc & * & \frac{cg}{f}\\
                f &  df & \frac{fe}{c} & *
                \end{pmatrix} = \begin{pmatrix}
                1 & 0 & 0 & 0\\
                0 & 1 & 0 & 0\\
                0 &  0 & c & 0 \\
                0 &  0  & 0 & f
                \end{pmatrix} \begin{pmatrix}
                * & a & c & f\\
                b & * & c & f\\
                1 &  d & * & g \\
                1 &  d  & e & *
                \end{pmatrix} \begin{pmatrix}
                1 & 0 & 0 & 0\\
                0 & 1 & 0 & 0\\
                0 &  0 & \frac{1}{c} & 0 \\
                0 &  0  & 0 & \frac{1}{f} 
                \end{pmatrix} .
\] Here, the matrix in the middle on the right hand side is $\tw(A,\{1,2\})$. Hence, in this case, $B\DE \tw(A,S)$.

Finally, the last case where $\left(1-\frac{bd}{h}\right)$ and $e-i$ are zero. Then,
\[B^T= \begin{pmatrix}
                * & \frac{a}{d} & c & f\\
                  bd & * & dc & df\\
                1 & 1 & * & g\\
                1 &  1 & e & *
                \end{pmatrix} = \begin{pmatrix}
                1 & 0 & 0 & 0\\
                0 & d & 0 & 0\\
                0 &  0 & 1 & 0 \\
                0 &  0  & 0 & 1
                \end{pmatrix} \begin{pmatrix}
                * & a & c & f\\
                b & * & c & f\\
                1 &  d & * & g \\
                1 &  d  & e & *
                \end{pmatrix} \begin{pmatrix}
                1 & 0 & 0 & 0\\
                0 & \frac{1}{d} & 0 & 0\\
                0 &  0 & 1 & 0 \\
                0 &  0  & 0 & 1 
                \end{pmatrix} .
\] Like the previous case, the matrix in the middle is $\tw(A,\{1,2\})$. Hence, in this case, $\tw(A,S)\DE B^T$. This completes the proof of \cref{lem:size4Matrix}.
\end{proof}
%\clearpage
%\appendix
%\section*{Appendix}

\end{document}